\documentclass{lmcs}
\pdfoutput=1

\usepackage{lastpage}
\lmcsdoi{22}{2}{32}
\lmcsheading{}{\pageref{LastPage}}{}{}%
{Jul.~01,~2025}{Jun.~23,~2026}{}

\usepackage[utf8]{inputenc}
\usepackage{float}

\usepackage{amsthm}
\usepackage[all]{xy}
\usepackage{amsmath}
\usepackage{amssymb}
\usepackage{amsfonts}
\usepackage{ebproof}
\usepackage{graphicx}
\usepackage{xcolor}
\usepackage{enumerate}

\newcommand{\moncut}[1]{$\mathsf{cut}$}

\newcommand{\mcut}{\mathsf{cut}}
\newcommand{\cut}{\mathsf{cut}}

\newcommand{\lsum}{\vee}
\newcommand{\lprod}{\wedge}
\newcommand{\lfun}{\rightarrow}
\newcommand{\ann}[1]{\mu^{#1}}
\newcommand{\nann}[1]{\nu^{#1}}
\newcommand{\md}{\;\mathsf{mod}\;}

\newcommand{\term}[1]{\mathbf{#1}}

\NewDocumentCommand{\termSum}{O{}}{ \term{\sumSym} {#1} }
\NewDocumentCommand{\termPlus}{m O{}}{ \term{\plusSym}\IfValueT{#2}{\mathrlap{^{\!#2}}}_{\!#1} }
\NewDocumentCommand{\termTimes}{o}{ \term{\prodSym} \IfValueT{#1}{^{#1\kern-1pt}} }
\NewDocumentCommand{\termResid}{o}{ \term{\resSym} \IfValueT{#1}{^{#1}} }
\NewDocumentEnvironment{scriptdisplay}{}
    {\begin{equation*}\scriptstyle}
    {\end{equation*}}

\NewDocumentEnvironment{ssprooftree}{}
    {\scalebox{0.7}{\begin{prooftree}[compact]}
    {\end{prooftree}}}

\NewDocumentEnvironment{ssprooftree*}{}
    {\begin{prooftree*}[compact]}
    {\end{prooftree*}}
\newcommand{\ReductionArrow}{\Rightarrow}
\DeclareDocumentCommand{\reduction}{ O{} m m }{%
    \scalebox{0.7}{\begin{prooftree}[#1] #2 \end{prooftree}}%
    \;\; \ReductionArrow \;\;%
    \scalebox{0.7}{\begin{prooftree}[#1] #3 \end{prooftree}}%
  }

\NewDocumentCommand \ntStyle { m }{ \mathsf{#1} }
\NewDocumentCommand \nt { s m m }{ \IfBooleanTF{#1}{\hat{\ntStyle{N}}}{\ntStyle{N}}^{#2}_{#3} }

\NewDocumentCommand{\setof}{ m o }{\{ \, #1 \IfValueT{#2}{ : #2}\, \} }
\NewDocumentCommand{\Setof}{ m o }{\Bigl\{ #1 \IfValueT{#2}{ : #2} \Bigr\} }
\NewDocumentCommand{\SETOF}{ m o }{\Biggl\{ #1 \IfValueT{#2}{ : #2} \Biggr\} }

\NewDocumentCommand{\Sub}{mo}{[ #1\IfNoValueF{#2}{/#2} ]}

\RenewDocumentCommand{\perp}{ m }{ \neg {#1} }
\newcommand{\plusSym}{{+}}

\newcommand{\prodSym}{{\times}}

\newcommand{\sumSym}{\Sigma}

\NewDocumentCommand{\rul}{ m e{_} o }{%
    \ensuremath{\mathsf{#1}%
    \IfValueT{#2}{_{#2}}%
    \IfValueT{#3}{^{#3}}
    }}
\NewDocumentCommand{\nrul}{ m o }{
    \ensuremath{\perp{\mathsf{#1}}
    \IfValueT{#2}{^{#2}}
    }}
\NewDocumentCommand{\ruleset}{ o }{\operatorname{\mathit{Rules}}\IfValueT{#1}{(#1)}}
\NewDocumentCommand{\var}{o}{\mathit{V}\IfValueT{#1}{\!(#1)}}

\newcommand{\rset}[1]{[\![ #1 ]\!]}
\newcommand{\cns}{\mathcal{O}}

\newcommand{\orda}{a}

\newcommand{\ordc}{c}

\makeatletter
\def\enddoc@text{%
  \vspace*{\fill}%
  \noindent
  \hfill\begin{minipage}{.58\textwidth}
  \vbox to 0pt{\vskip12pt%
    \fontsize{6}{7\p@}\normalfont\upshape
    \everypar{}%
    \noindent\fontencoding{T1}%
    \headertextsf{This work is licensed under the Creative Commons
    Attribution License. To view a copy of this license, visit
    \texttt{https://creativecommons.org/licenses/by/4.0/} or send a
    letter to Creative Commons, 171 Second St, Suite 300, San Francisco,
    CA 94105, USA, or Eisenacher Strasse~2, 10777 Berlin, Germany}\vss}
    \par
    \kern\z@
  \end{minipage}%
}
\makeatother

\title{Computation by infinite descent made explicit}
\author[S.~Enqvist]{Sebastian Enqvist}

\begin{document}

\begin{abstract}
We introduce a non-wellfounded proof system for intuitionistic logic extended with inductive and co-inductive definitions, based on a syntax in which fixpoint formulas are annotated with explicit variables for ordinals. We explore the computational content of this system, in particular we introduce a notion of computability and show that every valid proof is computable. As a consequence, we obtain a normalization result for proofs of what we call finitary formulas. A special case of this result is that every proof of a sequent of the appropriate form represents a unique function on natural numbers. Finally, we derive a categorical model from the proof system and show that least and greatest fixpoint formulas correspond to initial algebras and final coalgebras respectively.
\end{abstract}
\maketitle

\section{Introduction}

Non-wellfounded and cyclic proofs originated in the context of modal $\mu$-calculus with the search of a completeness proof for Kozen's axiomatization \cite{kozen1983results,niwinski1996games,walukiewicz2000completeness}. Since then non-wellfounded proof theory has grown into quite a large area of research. Besides modal $\mu$-calculus, non-wellfounded proofs have been implemented for many different logics involving induction or co-induction in one form or another, including provability logic \cite{shamkanov2014circular,savateev2021non}, arithmetic \cite{simpson2017cyclic,berardi2017equivalence,das2020logical,das2023cyclic} or first-order logic with inductively defined predicates \cite{brotherston2011sequent,berardi2019classical}. The core idea of non-wellfounded proof theory is that explicit rules or axioms of induction, which are required in the more traditional setting of wellfounded proofs, can be exchanged for a structural condition on non-wellfounded proof trees ensuring that the proof as a whole implicitly expresses a sound inductive argument. The advantage of working in this way is that complicated induction invariants implicit in the cyclic reasoning need not be made explicit in a formula. For the modal $\mu$-calculus this leads to finitary and analytic proof systems, where wellfounded proof systems either use infinitely branching rules or rely on cuts in an essential way. Indeed, a key part of the Kozen-Walukiewicz completeness proof for modal $\mu$-calculus consists of associating intricate induction invariants with branches in non-wellfounded tableaux. 

Beginning with the work of Santocanale \cite{santocanale2002calculus} and Fortier \cite{fortier2013cuts}, a lot of research has approached non-wellfounded proofs from the perspective of Curry-Howard correspondence, i.e. viewing proofs as programs (or terms) and formulas as specifications (or types) \cite{sorensen2006lectures}. Research in this area has yielded a wealth of interesting results on the computational expressivity of non-wellfounded proofs. In particular, Das has shown that a cyclic version of G\"{o}del's System $\mathbf{T}$ types the same computable functions as the usual wellfounded presentation, i.e. precisely the provably recursive functions of $\mathbf{PA}$ \cite{das2020circular}. The functions that can be captured by cyclic proofs for intuitionistic logic extended by arbitrary least and greatest fixpoints has also been characterized, again yielding equivalence with the wellfounded presentation \cite{curzi2023computational}. Non-wellfounded proofs for least and greatest fixpoints have also been considered in the context of linear logic, extending previous work using wellfounded proofs \cite{baelde2012least}, leading in particular to powerful cut-elimination theorems \cite{baelde2016infinitary,baelde2022bouncing,saurin2023linear}.

To specify a system of non-wellfounded proofs, one generally needs to formulate a criterion of \emph{correctness}, separating valid proofs from invalid derivation trees. If proofs are seen as ``truth certificates'' then correctness is simply meant to entail that the proved formula is true. In the setting of Curry-Howard correspondence, correctness rather means that the proof as a program terminates on all inputs, i.e. a form of normalization or cut elimination holds. These criteria come in a few different guises, but the basic idea is the same and can be traced back to Fermat's method of ``proof by infinite descent'', a form of contrapositive argument. For arithmetic, this means that falsity of the conclusion would entail an infinitely descending chain of natural numbers. For the more general case of inductive definitions and fixpoints, it means more generally that falsity of the conclusion would entail an infinitely descending chain of ordinals, corresponding to the approximants of the least and greatest fixpoints as characterized by the Knaster-Tarski theorem. The most traditional correctness criterion is the \emph{trace condition}, stating that on every infinite branch of a non-wellfounded proof tree, there should be a trace (also called \emph{thread}) of formulas on which some least fixpoint is unfolded infinitely many times on the left-hand side of the sequent. With this kind of correctness criterion, the infinite descent argument is implicit, with no explicit mention of ordinals. A disadvantage with this sort of correctness criterion is that information about the descent of ordinal approximants is tied to individual formula occurrences in the proof tree, and as a consequence the condition interacts poorly with cuts. In other words there is a problem of \emph{compositionality}, in that the correctness criterion puts restraints on how proofs can be composed by cuts. To deal with this issue, a correctness condition based on a more general form of traces called ``bouncing threads'' has recently been proposed \cite{baelde2022bouncing}.  

As an alternative to the trace-based validity condition, non-wellfounded proof systems with explicit variables for ordinals were proposed by Dam and Gurov for the modal $\mu$-calculus  \cite{dam2002mu} and further explored in the context of first-order fixpoint logic in \cite{sprenger2003structure,sprenger2003global} and more recently \cite{afshari2024cyclic}. Similar ideas have appeared in type theory in connection with reasoning about coinductive types and co-patterns, see \cite{abel2013wellfounded}. Recently, cyclic proofs with ordinal variables have been used to provide a complete, finitary proof system for the two-way modal $\mu$-calculus \cite{afshari2023proof}. None of these papers are however concerned with Curry-Howard correspondence for non-wellfounded proofs. 

The advantage of this type of proof system is that the correctness condition gets disentangled from the structure of traces of syntactic formula occurrences in the proof tree.  As a result, one might expect better interaction with the cut rule. The starting observation of this paper is that, in fact, precisely the sort of problematic examples presented in \cite{baelde2022bouncing} as a motivation for the bouncing thread condition can be captured in a rather transparent way using  ordinal variables; the research presented in this paper started out with an attempt to better understand the semantics of the bouncing thread condition.  Besides being arguably more transparent in terms of the intended meaning of the correctness criterion, there is a practical advantage too: while it was observed in \cite{baelde2022bouncing} that the bouncing thread validity condition is undecidable, the correctness criterion we consider here is known to be decidable \cite{sprenger2003global}. 

In this paper we explore Dam-Gurov-Sprenger style non-wellfounded (sequent calculus) proofs with ordinal variables from a Curry-Howard perspective, where our base logic is intuitionistic propositional logic with inductive and co-inductive definitions given by least and greatest fixpoints. Thus our system is a fragment of the $\mu$-calculus over intuitionistic logic. For background on the computational content of this system, see \cite{clairambault2009least}.  We provide a semantics of our proof system in terms of a notion of computability, and show that all valid proofs are computable in this sense. As a corollary we get a normalization theorem for proofs of what we call ``finitary'' formulas, capturing finite data types like natural numbers, lists and finite binary trees. In particular, we obtain the usual result that any proof of a sequent of the form $N,\hdots, N \vdash N$, where $N$ is the formula corresponding to the natural numbers type via Curry-Howard and with $k$ copies of $N$ on the left-hand side, computes a unique $k$-ary function on natural numbers. We also show how the computability semantics gives rise to a categorical model in which least and greatest fixpoints are initial algebras and final coalgebras respectively.    

The aim of this paper is to contribute to the project, initiated by Baelde et al. in \cite{baelde2022bouncing}, of searching  for natural validity criteria for non-wellfounded proofs that extend the class of valid proofs to include computationally meaningful proofs beyond those recognized by the standard progressing thread condition. It can be viewed as a pilot study, proposing the Dam-Gurov-Sprenger framework of explicit ordinal variables as a suitable framework for this purpose. In support of this we observe two key facts: first, like bouncing threads, our validity condition does capture some proofs that are not recognized as valid according to the standard progressing thread condition. Second, by providing a computability semantics, we demonstrate that all valid proofs in our sense are indeed computationally meaningful. That is, valid proofs represent well-defined programs. Our hope is that this motivates further exploration of this type of proof system from the Curry-Howard perspective, with many questions remaining. In particular, we leave for future research a more detailed comparison with the bouncing thread validity condition - can every bouncing thead valid proof be represented using ordinal variables, and vice versa, does every valid proof using ordinal variables represent a bouncing thread valid proof? 

\section{Proof system}

In this section we introduce the proof system that we will study.

\subsection{Formulas}

We fix an infinite supply of proposition variables, written with capital letters $X,Y,Z,\hdots$. 
The language $\mathcal{L}_{\mu}$ of intuitionistic propositional logic extended with least and greatest fixpoints is given by the following grammar:
\[ \mathcal{L}_\mu \ni A :=  \top \mid A \lsum A \mid A \lprod A \mid  A \lfun A \mid \nu X. B \mid \mu X. B\]
In the formation rules for $\mu$ and $\nu$, the variable $X$ is required to have only positive free occurrences in $B$, where positive and negative subformula occurrences are defined as usual. In particular, a positive subformula occurrence in $A \lfun B$ is either a positive subformula occurrence in $B$ or a negative subformula occurrence in $A$, and negative subformula occurrence in $A \lfun B$ is either a negative subformula occurrence in $B$ or a positive subformula occurrence in $A$.  Free and bound variables of a formula are defined as usual, and formulas with no free variables are called sentences. 

We shall consider a restricted version of the syntax, in which we make use only of a single fixpoint variable $X$ and restrict free variable occurrences to be \emph{strictly positive}, meaning no free variable occurrences in the antecedent of any implication. This simplifies some matters, in particular it is straightforward to prove that such strictly positive formulas give rise to functors on the ``category of proofs'' that we shall consider. This language can be called the extension of intuitionistic propositional logic with \emph{inductive and co-inductive definitions}, and will be denoted by $\underline{\mathcal{L}}_{\mathsf{ID}}$. The grammar is as follows:
\[\underline{\mathcal{L}}_{\mathsf{ID}} \ni A:= \top  \mid A \lsum A \mid A \lprod A \mid A \lfun A \mid \mu X.B \mid \nu X. B\]
\[\underline{\mathcal{L}}^+(X) \ni B := A \mid X \mid \top  \mid B \lsum B \mid B \lprod B \mid A \to B\]

From a Curry-Howard perspective, formulas in this language represent data types. Conjunction and disjunction represent product and co-product type constructors, implications give function types and the binders $\mu$ and $\nu$ provide least and greatest fixpoints of type constructions depending on some variable $X$. Some typical examples are:
\begin{itemize}
\item Natural numbers, represented by $\mu X. \top \vee X$. We abbreviate this by $N$.
\item Lists (of natural numbers), represented by $\mu X. \top \vee (N \wedge X)$. We abbreviate this by $L$.
\item Infinite streams (of natural numbers), represented by $\nu X. N \wedge X$. We abbreviate this by~$S$. 
\item Wellfounded $\omega$-branching trees, represented by $\mu X. \top \vee (N \to X)$.
\end{itemize}
Note that these examples all belong to the language $\underline{\mathcal{L}}_{\mathsf{ID}}$ of inductive and co-inductive definitions that we consider here.  

\subsection{Ordinal variables, constraints and sequents}

The Knaster-Tarski theorem \cite{knaster1928theoreme,tarski1955lattice} characterizes the least fixpoint of a monotone function $f$ on a complete lattice as the join of its ordinal approximants ``from below'', starting from the bottom element and transfinitely iterating the function $f$. Similarly the greatest fixpoint is the meet of its ordinal approximants ``from above''. In this paper we will explore a proof system based on an extension of the syntax with explicit variables referring to ordinal approximants, pioneered by Dam, Gurov and Sprenger \cite{dam2002mu,sprenger2003structure,sprenger2003global}. 

We shall make use of two types of terms for reasoning about ordinals: \emph{variables}, which will be denoted by lower-case letters $a,b,c,\hdots$, and a special constant $\infty$. We can think of the latter as standing for an arbitrarily chosen ordinal that is at least as large as any of the closure ordinals of fixpoints that we consider, where the \emph{closure ordinal} of a least or greatest fixpoint is the smallest ordinal at which its series of ordinal approximants stabilizes. Ordinal terms will be denoted by lower-case greek letters $\alpha,\beta,\gamma,\hdots$, where an ordinal term is either an ordinal variable or the constant $\infty$. 

We will make use of minimal means for reasoning about ordinals via inequalities $\alpha < \beta$, where $\alpha,\beta$ are ordinal terms. 
\begin{defi}
An \emph{ordinal constraint} is a finite rooted tree whose vertices are labelled injectively by ordinal terms, the root is labelled by the constant $\infty$ and all other vertices are labelled by ordinal variables. Given an ordinal constraint $\cns$ and ordinal terms $\alpha$,$\beta$ we write $\cns \Vdash \alpha < \beta$ if either $\beta = \infty$ or $\alpha\neq \beta$ and $\alpha$ is a descendant of $\beta$ in $\cns$. Note that, in particular, we have $\cns \Vdash \infty < \infty$ for any constraint $\cns$. The \emph{trivial} constraint is the one consisting only of the root labelled $\infty$.  
\end{defi}

 To make use of ordinal terms to reason about fixpoints, we introduce an extended version of the language presented in the previous section. The extension is given by the following grammar:
\[\mathcal{L}_{\mathsf{ID}} \ni A:= \top \mid A \lsum A \mid A \lprod A \mid A \lfun A \mid \mu^\alpha X.B \mid \nu^\alpha X. B \mid \exists a < \alpha : A \mid \forall a < \alpha: A\]
\[\mathcal{L}^+(X) \ni B := A \mid X \mid \top \mid B \lsum B \mid B \lprod B \mid A \to B\]

Here $\alpha$ ranges over ordinal terms. We impose a side constraint on the formation of formulas $\eta^a X.B$ for $\eta \in \{\mu,\nu\}$ to the effect that the ordinal variable $a$ may not occur free in $B$. Without this constraint we could form a formula like:
\[\nu^a X. (\top \wedge (X \wedge  \mu^a X.\top \vee X)\]
Here, the variable $a$ simultaneously refers to an approximant of two different fixpoints, which is a complication we wish to avoid. 
  Note that we have also included bounded quantification over ordinal variables, following \cite{afshari2023proof,afshari2024cyclic}.  Intuitively, given an ordinal term $\alpha$, $\mu^\alpha X. B$ represents the $\alpha$-th approximant of the least fixpoint, and $\nu^\alpha X. B$ represents the $\alpha$-th approximant of the greatest fixpoint. The actual least fixpoint can be identified with $\mu^\infty X. B$ and the greatest fixpoint with $\nu^\infty X. B$, so the language presented in the previous section is a fragment of $\mathcal{L}_{\mathsf{ID}}$ via the trivial translation that just annotates all fixpoints with $\infty$. A \emph{pure} formula of the language $\mathcal{L}_{\mathsf{ID}}$ is defined to be a formula in which every fixpoint binder is annotated by $\infty$, and in which there are no occurrences of the bounded ordinal quantifiers $\exists a < \alpha$ or $\forall a < \alpha$. A \emph{closed} formula is one in which no occurrences of ordinal variables are free. The definition of a free occurrence of an ordinal variable is as expected; for example in the formula $\forall a < b : (\mu^b X. X \to  \nu^a X. X)$, the variable $a$ only occurs bound while the two occurrences of $b$ are free.

A \emph{sequent} $\cns,\Gamma \vdash A $ consists of a constraint $\cns$, a multiset of formulas $\Gamma$ and a formula $A$, where we require that all ordinal variables occurring in $\Gamma$ or $A$ also occur in the constraint. A \emph{pure sequent} is a sequent in which the constraint is trivial and all formulas appearing in the sequent are pure.  As a convention, pure sequents will be written as $\Gamma \vdash A$. When the constraint is not trivial, we will write it as a list of inequalities $\alpha < \beta$ where the constraint contains vertices $u,v$ such that $u$ is labelled by $\beta$, $v$ is a child of $u$ and is labelled by $\alpha$.

Given a constraint $\cns$ containing a (unique) vertex $u$ labelled $\alpha$, and $c$ an ordinal variable that does not occur in $\cns$, we denote by $\cns+_\alpha c$ the constraint in which a new vertex $v$ is added as a child of $u$ and labelled by $c$. We shall sometimes write $\cns(\alpha < \beta)$ to denote a constraint $\cns$ such that $\cns \Vdash \alpha < \beta$.

\subsection{Proof rules}

We now give the rules of our proof system. First we include two axioms, one for $\top$ and an identity axiom for each formula $A$:
\[
\scalebox{0.7}{\begin{prooftree}
\hypo{}
\infer1[$\mathsf{\top}$]{\cns \vdash \top}
\end{prooftree}}
\quad\quad
\scalebox{0.7}{\begin{prooftree}
\hypo{}
\infer1[$id$]{\cns,A \vdash A}
\end{prooftree}}
\]

Remaining rules of our proof system are displayed in Figures \ref{f:proof-rules} and $\ref{f:struct-rules}$.

\begin{figure}
\begin{center}
\[
\begin{tabular}{c@{\qquad}  c@{\qquad}}
\(
\scalebox{0.7}{\begin{prooftree}
\hypo{\cns, \Gamma \vdash   A_i}
\infer1[$\lsum_R$]{\cns, \Gamma \vdash A_1 \lsum A_2}
\end{prooftree}}
\)
&
\(
\scalebox{0.7}{\begin{prooftree}
\hypo{\cns, \Gamma,  A_1 \vdash B}
\hypo{\cns, \Gamma, A_2 \vdash B}
\infer2[$\lsum_L$]{\cns, \Gamma, A_1 \lsum A_2 \vdash   B}
\end{prooftree}}
\)\\[1.5em]
\(
\scalebox{0.7}{\begin{prooftree}
\hypo{\cns, \Gamma \vdash A}
\hypo{\cns, \Delta \vdash  B}
\infer2[$\lprod_R$]{\cns, \Gamma, \Delta \vdash A \lprod B}
\end{prooftree}}
\)
&
\(
\scalebox{0.7}{\begin{prooftree}
\hypo{\cns, \Gamma,  A, B \vdash C}
\infer1[$\lprod_L$]{\cns, \Gamma, A \lprod B \vdash C}
\end{prooftree}}
\)
\\[1.5em]
\(
\scalebox{0.7}{\begin{prooftree}
\hypo{\cns, \Gamma, A \vdash B}
\infer1[$\lfun_R$]{\cns, \Gamma \vdash A \lfun B}
\end{prooftree}}
\)
&
\(
\scalebox{0.7}{\begin{prooftree}
\hypo{\cns, \Gamma \vdash A}
\hypo{\cns,  \Delta, B \vdash C}
\infer2[$\lfun_L$]{\cns,  \Gamma ,\Delta, A \lfun B  \vdash C}
\end{prooftree}}
\)
\\[1.5em]
\(
\scalebox{0.7}{\begin{prooftree}
\hypo{\cns(\beta < \alpha), \Gamma \vdash  A[\mu^\beta X. A/X]}
\infer1[$\mu_R$]{\cns(\beta < \alpha), \Gamma \vdash \mu^\alpha X. A}
\end{prooftree}}
\)
&
\(
\scalebox{0.7}{\begin{prooftree}
\hypo{\cns +_\alpha b,  \Gamma , A[\mu^b X. A/X] \vdash  B}
\infer1[$\mu_L$]{\cns,  \Gamma ,\mu^\alpha X. A\vdash B}
\end{prooftree}}
\)
\\[1.5em]
\(
\scalebox{0.7}{\begin{prooftree}
\hypo{\cns +_\alpha b, \Gamma \vdash  A[\nu^b X.A/X] }
\infer1[$\nu_R$]{\cns, \Gamma \vdash \nu^\alpha X. A}
\end{prooftree}}
\)
&
\(
\scalebox{0.7}{\begin{prooftree}
\hypo{\cns(\beta < \alpha),  \Gamma ,A[\nu^\beta X. A/X]\vdash B}
\infer1[$\nu_L$]{\cns(\beta < \alpha), \Gamma, \nu^\alpha X. A \vdash  B}
\end{prooftree}}
\)
\\[1.5em]
\(
\scalebox{0.7}{\begin{prooftree}
\hypo{\cns(\beta < \alpha), \Gamma \vdash  A[\beta/c]}
\infer1[$\exists_R$]{\cns(\beta < \alpha), \Gamma \vdash \exists c < \alpha: A}
\end{prooftree}}
\)
&
\(
\scalebox{0.7}{\begin{prooftree}
\hypo{\cns +_\alpha b,  \Gamma , A[b/c] \vdash  B}
\infer1[$\exists_L$]{\cns,  \Gamma ,\exists c < \alpha: A \vdash B}
\end{prooftree}}
\)
\\[1.5em]
\(
\scalebox{0.7}{\begin{prooftree}
\hypo{\cns +_\alpha b, \Gamma \vdash  A[b/c] }
\infer1[$\forall_R$]{\cns, \Gamma \vdash \forall c < \alpha: A}
\end{prooftree}}
\)
&
\(
\scalebox{0.7}{\begin{prooftree}
\hypo{\cns(\beta < \alpha),  \Gamma ,A[\beta/c]\vdash B}
\infer1[$\forall_L$]{\cns(\beta < \alpha), \Gamma, \forall c < \alpha: A \vdash  B}
\end{prooftree}}
\)
\end{tabular}
\]
\caption{
\label{f:proof-rules} Logical rules}
\end{center}
\end{figure}

\begin{figure}
\[
\begin{tabular}{c}
\begin{tabular}{c@{\qquad}  c@{\qquad}}
\(\scalebox{0.7}{\begin{prooftree}
\hypo{\cns, \Gamma \vdash B}
\infer1[$W$]{\cns, A, \Gamma \vdash B}
\end{prooftree}}
\)
&
\(
\scalebox{0.7}{\begin{prooftree}
\hypo{\cns, A, A, \Gamma \vdash B}
\infer1[$C$]{\cns, A, \Gamma \vdash B}
\end{prooftree}}
\)
\end{tabular}
\\[1.5em]
\(
\scalebox{0.7}{\begin{prooftree}
\hypo{\cns, \Gamma_1 \vdash A_1}
\hypo{\hdots}
\hypo{\cns, \Gamma_n \vdash A_n}
\hypo{\cns,  A_1,\hdots, A_n, \Delta \vdash  B}
\infer4[$cut$]{\cns, \Gamma_1,\hdots,\Gamma_n, \Delta  \vdash  B}
\end{prooftree}}
\)
\end{tabular}
\]
\caption{\label{f:struct-rules} Structural rules}
\end{figure}

In the rules ($\mu_L$) and ($\nu_R$), the variable $b$ is subject to the condition that it must be fresh, i.e. it may not occur in the conclusion, and likewise for the rules ($\exists_L$) and ($\forall_R$). The cut rule is formulated in the version that is sometimes called ``multicut'', allowing several cut formulas at once. This is quite common and convenient when working with non-wellfounded proofs. 

A \emph{derivation} is a (possibly infinite) tree in which vertices are appropriately labelled by sequents and proof rules. Note that since there is no weakening rule for constraints, an infinite branch $u_0,u_1,u_2,\hdots$ induces an infinite sequence $\cns_0,\cns_1,\cns_2,\hdots$ of constraints such that $\cns_i$ embeds into $\cns_j$ for $i < j$. We say that the inequality $\alpha < \beta$ \emph{eventually holds} on the branch if there is some $i$ for which $\cns_i \Vdash \alpha < \beta$.
\begin{defi}
A derivation is said to be \emph{valid} if, for every infinite branch $B$, there is an infinite sequence $a_0,a_1,a_2,\hdots$ of ordinal variables such that for each $i$, the inequality $a_{i+1} < a_i$ eventually holds on $B$. We refer to such a sequence as an infinite descending chain of ordinal variables.  Valid derivations will also be called \emph{proofs}. 
\end{defi}
Note the restriction to ordinal variables here; of course, we do not want to count the infinite chain $\infty > \infty > \infty \dotsm$ as an infinite descending chain, since that would trivialize the validity condition. 

A special role will be played by what we call \emph{closed proofs}:
\begin{defi}
A \emph{closed sequent} is a sequent of the form $\vdash A$, where $A$ is a closed formula, the constraint is trivial and the left-hand side is empty. A derivation or proof is called closed if its end sequent is closed. 
\end{defi}

Proofs are generally infinite structures. It is common in non-wellfounded proof theory to introduce a notion of \emph{cyclic proof}, which is a finite structure that represents an infinite but regular proof. The presence of ordinal variables makes the issue of how to define the notion of cyclic proof precisely somewhat more subtle \cite{sprenger2003global,afshari2024cyclic}. As this is not our focus here, we will not give an explicit definition of cyclic proofs and will work with infinite proof trees throughout. In practice, when we present infinite proofs by finite cyclic proof graphs, we trust it will be clear what infinite proof is represented by the finite presentation. For example, consider the following proof graph:
\[
\scalebox{0.7}{\begin{prooftree}
\hypo{}
\infer1[$\mathsf{ax}_\top$]{a< \infty\vdash  \top}
\hypo{a < \infty  \vdash \nu^a X. \top \wedge X\;\dagger}
\infer2[$\wedge_R$]{a < \infty \vdash \top \wedge \nu^a X. \top \wedge X}
\infer1[$\nu_R$]{\vdash \nu^\infty X. \top \wedge X \; \dagger}
\end{prooftree}}
\]

Here, we used the symbol ``$\dagger$'' to mark a cycle in the proof, and the intended meaning of this is that we want to represent the infinite proof obtained by continuing from the marked leaf as from the associated root, treating  the variable $a$ at the leaf as $\infty$ at the root. Explicitly, the infinite proof represented is:
\[
\scalebox{0.7}{\begin{prooftree}
\hypo{}
\infer1[$\mathsf{ax}_\top$]{ a< \infty\vdash  \top}
\hypo{}
\infer1[$\mathsf{ax}_\top$]{b < a< \infty\vdash  \top}
\hypo{}
\ellipsis{}{c < b < a < \infty \vdash \top \wedge \nu^c X. \top \wedge X}
\infer1[$\nu_R$]{b < a < \infty \vdash \nu^b X. \top \wedge X}
\infer2[$\wedge_R$]{b < a < \infty \vdash \top \wedge \nu^b X. \top \wedge X}
\infer1[$\nu_R$]{a < \infty \vdash \nu^a X. \top \wedge X}
\infer2[$\wedge_R$]{a < \infty \vdash \top \wedge \nu^a X. \top \wedge X}
\infer1[$\nu_R$]{\vdash \nu^\infty X. \top \wedge X}
\end{prooftree}}
\]

\subsection{Example: bouncing threads}

In their work on the ``bouncing threads'' validity condition for non-wellfounded proofs \cite{baelde2022bouncing}, Baelde et al. provide the following example proof:
\[
\scalebox{0.7}{\begin{prooftree}
\hypo{}
\infer1[$id$]{N \vdash N}
\hypo{ S \vdash N \lprod S \; \dagger}
\infer1[$\nu_R$]{S \vdash S}
\infer2[$\wedge_R$]{N, S \vdash N \lprod S}
\infer1[$\wedge_L$]{N \lprod S \vdash \underline{N \lprod S}}
\hypo{}
\infer1[$id$]{N \vdash N}
\infer1[$\lsum_R$]{N \vdash \top \lsum N}
\infer1[$\mu_R$]{N \vdash N}
\hypo{}
\infer1[$id$]{S \vdash S}
\infer2[$\wedge_R$]{N, S \vdash N \lprod S}
\infer1[$\wedge_L$]{\underline{N \lprod S} \vdash N \lprod S}
\infer2[$cut$]{N \lprod S \vdash N \lprod S}
\infer1[$\nu_L$]{S \vdash N \lprod S}
\infer1[$W$]{N,S \vdash N \lprod S}
\infer1[$\wedge_L$]{N\lprod S \vdash N \lprod S}
\infer1[$\nu_L$]{ S \vdash N \lprod S \; \dagger}
\infer1[$\nu_R$]{ S \vdash S}
\end{prooftree}}
\]

Here, we recall that $N := \mu X. \top \vee X$ represents the type of natural numbers and $S := \nu X. N \wedge X$ represents the type of infinite streams of natural numbers. Viewed as a program, this proof thus computes a function from streams to streams. What this program does is that it removes every second element of the stream, and also increments every element that it keeps. In other words, given a stream $n_0n_1n_2\dotsm$ as input, the proof computes the stream $(n_1+1)(n_3+1)(n_5+1)\dotsm$ as output. The question is what it is about the proof that guarantees that, for a given stream as input, it always produces an infinite stream as output. The standard condition used to ensure this is that every infinite branch of the proof tree has a progressive thread, i.e. an infinite sequence of formulas connected by the derivation steps on which either a (highest ranking) greatest fixpoint is unfolded infinitely often on the right-hand side of the sequent, or a (highest ranking) least fixpoint is unfolded infinitely often on the left-hand side. But this condition fails for the proof above, essentially due to the cut disconnecting the right-hand trace on the cycle, so a more general validity condition is needed. The response in \cite{baelde2022bouncing} is to say that there is in fact a progressing trace in the proof, of a more general kind: it starts from the lowest fixpoint unfolding on the right, and ``bounces'' against the axioms in the right subtree, back down through the underlined cut formulas until it reaches the  end of the cycle in the left subtree.

Now we consider what this proof looks like when formulated using ordinal variables. Here, we take as a convention that for a given ordinal term $\alpha$,  $N^\alpha$ abbreviates $\mu^\alpha X. \top \vee X$ and similarly $S^\alpha$ abbreviates $\nu^\alpha X. N \wedge X$. The proof becomes:
\[
\scalebox{0.7}{\begin{prooftree}
\hypo{}
\infer1[$id$]{a < \infty,N^\infty \vdash N^\infty}
\hypo{b < a < \infty, S^{\infty} \vdash N^\infty \lprod S^b \; \dagger}
\infer1[$\nu_R$]{a < \infty, S^{\infty} \vdash S^a}
\infer2[$\wedge_R$]{a < \infty, N^\infty, S^{\infty} \vdash N^\infty \lprod S^a}
\infer1[$\wedge_L$]{ a < \infty,N^\infty \lprod S^{\infty} \vdash \underline{N^\infty \lprod S^a}}
\hypo{}
\infer1[$id$]{a < \infty, N^\infty \vdash N^\infty}
\infer1[$\vee_R$]{ a < \infty, N^\infty \vdash \top \lsum N^\infty}
\infer1[$\mu_R$]{ a < \infty, N^\infty \vdash N^\infty}
\hypo{}
\infer1[$id$]{a < \infty, S^a \vdash S^a}
\infer2[$\wedge_R$]{ a < \infty, N^\infty, S^a \vdash N^\infty \lprod S^a}
\infer1[$\wedge_L$]{ a < \infty,\underline{N^\infty \lprod S^a} \vdash N^\infty \lprod S^a}
\infer2[$cut$]{ a < \infty,N \lprod S^{\infty} \vdash N^\infty \lprod S^a}
\infer1[$\nu_L$]{ a < \infty,S^{\infty} \vdash N^\infty \lprod S^a}
\infer1[$W$]{ a < \infty,N^\infty ,S^{\infty} \vdash N^\infty \lprod S^a}
\infer1[$\wedge_L$]{a < \infty, N^\infty \lprod S^{\infty} \vdash N^\infty \lprod S^a \;\dagger}
\infer1[$\nu_L$]{ a < \infty, S^{\infty} \vdash N^\infty \lprod S^a }
\infer1[$\nu_R$]{ S^{\infty} \vdash S^\infty}
\end{prooftree}}
\]
This is a valid proof, and the infinite descent argument is clearly visible; the unique infinite branch will have an infinite descending chain of ordinal variables $a > b > c \dotsm$. The reason that we can capture this proof is that the cut no longer interrupts the progression of the infinite descent argument; the information about descent to some lower ordinal approximant to the greatest fixpoint $S$ is explicitly present in the cut formula.

What the proof does is to construct an infinite stream by producing a natural number as head and, for any given approximant given as value $k$ for the ordinal variable $a$, a description of the tail up to length $k$. We can think of it as a game, where one player (Prover) has to produce a stream $s_1$ in response to a stream $s_0$ provided by the opposing player (Refuter). First, Prover commits to produce a natural number as head together with an approximant to the tail, and asks Refuter for a value $k$ for the variable $a$. Now Prover has committed to construct the tail $s_1'$ of the stream up to the bound $k$.  First, Prover asks for the head of the stream $s_0$, discards it and keeps the tail $s_0'$. This is in turn decomposed as a natural number $n$ and a tail $s_0''$. From this Prover produces an approximant to the output stream in two steps: in the left subtree, the procedure is recursively called on the tail $s_0''$ to produce the tail $s_1'$, and the value $n \cdot s_1'$ is then passed to the right subtree via the cut. This value approximates a stream up to the value $k$ picked by Refuter, and now the right subtree simply increments the value of $n$ to $n+1$. Each recursive call of the procedure forces Refuter to pick a smaller value for the length of the approximants that Prover has to provide, so eventually Prover wins the game. 

An important task for future work is to systematically compare the present proof system with the bouncing threads validity condition. In particular we leave it as an open question whether, for any given proof satisfying the bouncing thread validity condition, we can always find an equivalent (in some sense to be made precise) proof using ordinal variables. What is clear in any case is that we do need to include the bounded ordinal quantifiers $\exists a < \alpha$ and $\forall a < \alpha$ for this purpose, as simple counter-examples could be constructed otherwise. The following example was suggested by an anonymous reviewer:
\[
\scalebox{0.7}{ 
\begin{prooftree}
\hypo{\vdots}
\infer1{\vdash \nu X.X }
\hypo{}
\infer1[$id$]{ \nu X.X \vdash \nu X.X}
\infer1[$\nu_R$]{ \vdash \nu  X.X \vdash \nu X.X}
\infer2[$cut$]{ \vdash \nu  X.X }
\hypo{}
\infer1[$id$]{ \nu X.X \vdash \nu X.X}
\infer1[$\nu_R$]{ \nu X.X \vdash \nu X.X}
\infer2[$cut$]{\vdash \nu X. X }
\end{prooftree}
}
\]

The infinite branch in this proof tree has a progressing bouncing thread, bouncing on the identiy axiom in every right-hand subproof. As pointed out by the anonymous reviewer, there seems to be no way to represent this proof using only the unfolding rules for annotated fixpoints. Using the ordinal quantifiers however, the proof can be reproduced in the following way:
\[
\scalebox{0.7}{ 
\begin{prooftree}
\hypo{\vdots}
\infer1{b_1 < b_0 < \infty \vdash \nu^{b_1} X.X }
\infer1[$\forall_R$]{b_0 < \infty \vdash \forall a < b_0: \nu^a X.X}
\hypo{}
\infer1[$id$]{c < b_0 < \infty, \nu^c X.X \vdash \nu^c X.X}
\infer1[$\forall_L$]{c < b_0 < \infty, \forall a < b_0 : \nu^{a} X.X \vdash \nu^c X.X}
\infer1[$\nu_R$]{b_0 < \infty \vdash \forall a < b_0: \nu^{a}  X.X \vdash \nu^\infty X.X}
\infer2[$cut$]{b_0 < \infty \vdash \nu^{b_0} X.X }
\infer1[$\forall_R$]{\vdash \forall a < \infty: \nu^a X.X}
\hypo{}
\infer1[$id$]{c < \infty, \nu^c X.X \vdash \nu^c X.X}
\infer1[$\forall_L$]{c < \infty, \forall a < \infty : \nu^a X.X \vdash \nu^c X.X}
\infer1[$\nu_R$]{\forall a < \infty: \nu^a  X.X \vdash \nu^\infty X.X}
\infer2[$cut$]{\vdash \nu^\infty X. X }
\end{prooftree}
}
\]
The example illustrates nicely that the question of representability of bouncing thread proofs in our system is a non-trivial one. We cannot expect to capture all bouncing thread valid proofs just by decorating proof trees with ordinal variables in the right way. We have to make explicit whatever reasoning about ordinals is implicit in a bouncing thread proof; the question is just how much syntax for representing information about ordinals we need in order to do that. 

\section{Cut reductions and Curry-Howard}

In this section we will introduce cut reductions and begin to study the computational content of proofs with ordinal variables. It will be convenient for this purpose to introduce a term notation for proofs.

\subsection{Term notation}

We present our term notation by rephrasing our derivation rules as deriving typing judgments of the form:
 \[\cns, x_1 : B_1,\hdots,x_n : B_n  \vdash u : A\] where $\cns$ is a constraint, $x_1 : B_1,\hdots,x_n : B_n$ is a context of variable declarations and $u$ is a proof term. Formally, we view proof terms as infinite trees in which vertices are labelled by variables and derivation rules. We often abbreviate a context $x_1 : B_1,\hdots,x_n : B_n$ as $\vec{x} : \Gamma$ where $\vec{x} = (x_1,\hdots,x_n)$ and $\Gamma = (B_1,\hdots,B_n)$.  The term notation for the $\top$-axiom and identity axiom is as follows:
\[
\scalebox{0.7}{\begin{prooftree}
\hypo{}
\infer1{\cns \vdash \mathsf{ax}_\top : \top}
\end{prooftree}}
\quad\quad
\scalebox{0.7}{\begin{prooftree}
\hypo{}
\infer1{\cns, x : A \vdash \mathsf{id}(x) : A}
\end{prooftree}}\]

The type-theoretic reformulation of our sequent calculus rules are presented in Figures \ref{f:logic-rules-type} and $\ref{f:struct-rules-type}$. 

\begin{figure}
\begin{center}
\[
\begin{tabular}{c@{\qquad}  c@{\qquad}}
\(
\scalebox{0.7}{\begin{prooftree}
\hypo{\cns, \vec{x} : \Gamma \vdash u_i :  A_i}
\infer1{\cns, \vec{x} : \Gamma \vdash \vee^i_R(u_i) : A_1 \lsum A_2}
\end{prooftree}}
\)
&
\(
\scalebox{0.7}{\begin{prooftree}
\hypo{\cns, \vec{x} : \Gamma,  z_1 : A_1 \vdash u_1 : B}
\hypo{\cns, \vec{x} : \Gamma, z_2 : A_2 \vdash u_2 : B}
\infer2{\cns, \vec{x} : \Gamma, y : A_1 \lsum A_2 \vdash \vee_L(y,z_1,z_2,u_1,u_2) :  B}
\end{prooftree}}
\)\\[1.5em]
\(
\scalebox{0.7}{\begin{prooftree}
\hypo{\cns, \vec{x} : \Gamma \vdash u : A}
\hypo{\cns, \vec{y} : \Delta \vdash  v : B}
\infer2{\cns, \vec{xy} : \Gamma, \Delta \vdash \wedge_R(u,v) : A \lprod B}
\end{prooftree}}
\)
&
\(
\scalebox{0.7}{\begin{prooftree}
\hypo{\cns, \vec{x} : \Gamma,  z_1 : A, z_2 : B \vdash u : C}
\infer1{\cns, \vec{x} : \Gamma, y : A \lprod B \vdash \wedge_L(y,z_1,z_2,u) : C}
\end{prooftree}}
\)
\\[1.5em]
\(
\scalebox{0.7}{\begin{prooftree}
\hypo{\cns, \vec{x} : \Gamma, y : A \vdash u : B}
\infer1{\cns, \vec{x} : \Gamma \vdash {\lfun}_R(y,u) : A \lfun B}
\end{prooftree}}
\)
&
\(
\scalebox{0.7}{\begin{prooftree}
\hypo{\cns, \vec{x}_0 : \Gamma \vdash u:A}
\hypo{\cns, y : B, \vec{x}_1 : \Delta \vdash v : C}
\infer2{\cns, \vec{x}_0\vec{x}_1 : \Gamma, \Delta, z : A \lfun B \vdash {\lfun}_L(z,y,u,v) : C}
\end{prooftree}}
\)
\\[1.5em]
\(
\scalebox{0.7}{\begin{prooftree}
\hypo{\cns(\beta < \alpha), \vec{x} : \Gamma \vdash  u : A[\mu^\beta X. A/X]}
\infer1{\cns(\beta < \alpha), \vec{x} : \Gamma \vdash \mu_R(\alpha,\beta,u) : \mu^\alpha X. A}
\end{prooftree}}
\)
& 
\(
\scalebox{0.7}{\begin{prooftree}
\hypo{\cns +_\alpha b,  \vec{x} : \Gamma, z : A[\mu^b X. A/X] \vdash u :   B}
\infer1{\cns,  \vec{x} : \Gamma, y : \mu^\alpha X. A \vdash \mu_L(\alpha,b,y,z,u) : B}
\end{prooftree}}
\)
\\[1.5em]
\(
\scalebox{0.7}{\begin{prooftree}
\hypo{\cns +_\alpha b, \vec{x} : \Gamma \vdash  u : A[\nu^b X.A/X] }
\infer1{\cns, \vec{x} : \Gamma \vdash \nu_R(\alpha, b, u) : \nu^\alpha X. A}
\end{prooftree}}
\)
&
\(
\scalebox{0.7}{\begin{prooftree}
\hypo{\cns(\beta < \alpha),  \vec{x} : \Gamma, z :  A[\nu^\beta X. A/X] \vdash u : B}
\infer1{\cns(\beta < \alpha), \vec{x} : \Gamma, y : \nu^\alpha X. A \vdash \nu_L(\alpha,\beta,y,z,u) : B}
\end{prooftree}}
\)
\\[1.5em]
\(
\scalebox{0.7}{\begin{prooftree}
\hypo{\cns(\beta < \alpha), \vec{x} : \Gamma \vdash  u : A[\beta/c]}
\infer1{\cns(\beta < \alpha), \vec{x} : \Gamma \vdash \exists_R(\alpha,\beta,u) : \exists c < \alpha: A}
\end{prooftree}}
\)
& 
\(
\scalebox{0.7}{\begin{prooftree}
\hypo{\cns +_\alpha b,  \vec{x} : \Gamma, z : A[b/c] \vdash u :   B}
\infer1{\cns,  \vec{x} : \Gamma, y : \exists c < \alpha : A \vdash \exists_L(\alpha,b,y,z,u) : B}
\end{prooftree}}
\)
\\[1.5em]
\(
\scalebox{0.7}{\begin{prooftree}
\hypo{\cns +_\alpha b, \vec{x} : \Gamma \vdash  u : A[b/c] }
\infer1{\cns, \vec{x} : \Gamma \vdash \forall_R(\alpha, b, u) : \forall c < \alpha: A}
\end{prooftree}}
\)
&
\(
\scalebox{0.7}{\begin{prooftree}
\hypo{\cns(\beta < \alpha),  \vec{x} : \Gamma, z :  A[\beta/c] \vdash u : B}
\infer1{\cns(\beta < \alpha), \vec{x} : \Gamma, y : \forall c < \alpha: A \vdash \forall_L(\alpha,\beta,y,z,u) : B}
\end{prooftree}}
\)
\end{tabular}
\]
\caption{
\label{f:logic-rules-type} Logical rules}
\end{center}
\end{figure}

\begin{figure}
\[
\begin{tabular}{c}
\begin{tabular}{c@{\qquad}  c@{\qquad}}
\(\scalebox{0.7}{\begin{prooftree}
\hypo{\cns, \vec{x} : \Gamma \vdash u : B}
\infer1{\cns, \vec{x} : \Gamma,  y : A \vdash W(y,u) : B}
\end{prooftree}}
\)
&
\( 
\scalebox{0.7}{\begin{prooftree}
\hypo{\cns, \vec{x} : \Gamma, z_0 : A, z_1 : A \vdash u : B}
\infer1{\cns, \vec{x} : \Gamma, y : A \vdash C(y,z_0,z_1,u) : B}
\end{prooftree}}
\)
\end{tabular}
\\[1.5em]
\(
\scalebox{0.7}{\begin{prooftree}
\hypo{\cns, \vec{x}_1 : \Gamma_1 \vdash u_1 : A_1}
\hypo{\hdots}
\hypo{\cns, \vec{x}_n : \Gamma_n \vdash u_n : A_n}
\hypo{\cns,  z_1 : A_1,\hdots, z_n : A_n, \vec{y} : \Delta \vdash v:  B}
\infer4{\cns, \vec{x}_1\dotsm\vec{x}_{n}\vec{y} : \Gamma_1,\hdots,\Gamma_n, \Delta  \vdash \mcut(u_1,\hdots,u_n,z_1, \hdots, z_n, v) :  B}
\end{prooftree}}
\)
\end{tabular}
\]
\caption{\label{f:struct-rules-type} Structural rules}
\end{figure}

\subsection{Substitution of ordinal variables}

Since left and right rules for fixpoints behave like quantifier rules, with the right $\nu$-rule and left $\mu$-rule respectively introducing eigenvariables, it is clear that cut reductions for fixpoint rules will involve some form of substitution for ordinal variables. It turns out that for our purposes we will only need a very restricted form of substitution: given a constraint $\cns$ and an ordinal variable $a$, we need to be able to perform the substitution of $\infty$ for $a$ provided that $a$ is a child of $\infty$. That is, the substitution is defined when there is no ordinal variable $b$ with $\cns \vdash a < b$. Assuming that the end sequent of a proof $u$ satisfies this constraint, we define the substitution $u[\infty/a]$ coinductively as in Figure \ref{f:substitution}. We include only the cases for fixpoint rules since the cases for quantifier rules are analogous and, in all other cases, the substitution simply commutes with the last proof rule applied. Note how the condition on the end sequent is maintained as an invariant, and guarantees that the cases listed are the only possible ones.  Any reader who feels uneasy about the soundness of a coinductive definition presented in this style is advised to consult \cite{kozen2017practical}.
\begin{figure}
\[
\begin{tabular}{c c c}
\( \left(
\scalebox{0.7}{\begin{prooftree}
\hypo{u_0}
\infer1{\cns+_a b, \Gamma \vdash A[\nu^b X. A/X]}
\infer1{\cns, \Gamma \vdash \nu^a X. A}
\end{prooftree}}
\right)[\orda := \infty] \) & \(=\) & \(
\scalebox{0.7}{\begin{prooftree}
\hypo{u_0[a := \infty]}
\infer1{\cns[\infty/a]+_\infty b, \Gamma[\infty/a] \vdash A[\infty/a][\nu^b X. A[\infty/a]/X]}
\infer1{\cns[\infty/a], \Gamma[\infty/a] \vdash \nu^\infty X. A[\infty/a]}
\end{prooftree}}
\) 
\\[2.0em]
\( \left(
\scalebox{0.7}{\begin{prooftree}
\hypo{u_0}
\infer1{\cns+_\alpha b, \Gamma \vdash A[\nu^b X. A/X]}
\infer1{\cns, \Gamma \vdash \nu^\alpha X. A}
\end{prooftree}}
\right)[\ordc := \infty] \) & \(=\) & \(
\scalebox{0.7}{\begin{prooftree}
\hypo{u_0[c := \infty]}
\infer1{\cns[\infty/c]+_\alpha b, \Gamma[\infty/c] \vdash A[\infty/c][\nu^b X. A[\infty/c]/X]}
\infer1{\cns[\infty/c], \Gamma[\infty/c] \vdash \nu^\alpha X. A[\infty/c]}
\end{prooftree}}
\) 
\\[2.0em]
\( \left(
\scalebox{0.7}{\begin{prooftree}
\hypo{u_0}
\infer1{\cns, \Gamma, A[\nu^a X.A/X] \vdash B}
\infer1{\cns, \Gamma, \nu^\infty X. A \vdash  B}
\end{prooftree}}
\right)[\orda := \infty] \) & \(=\) & \(
\scalebox{0.7}{\begin{prooftree}
\hypo{u_0[a := \infty]}
\infer1{\cns[\infty/a], \Gamma[\infty/a], A[\infty/a][ \nu^\infty X. A[\infty/a]/X] \vdash B[\infty/a]}
\infer1{\cns[\infty/a], \Gamma[\infty/a], \nu^\infty X. A[\infty/a] \vdash B[\infty/a]}
\end{prooftree}}
\) 
\\[2.0em]
\( \left(
\scalebox{0.7}{\begin{prooftree}
\hypo{u_0}
\infer1{\cns(\alpha < \beta), \Gamma, A[\nu^\alpha X.A/X] \vdash B}
\infer1{\cns(\alpha < \beta), \Gamma, \nu^\beta X. A \vdash  B}
\end{prooftree}}
\right)[\ordc := \infty] \) & \(=\) & \(
\scalebox{0.7}{\begin{prooftree}
\hypo{u_0[c := \infty]}
\infer1{\cns[\infty/c](\alpha < \beta), \Gamma[\infty/c], A[\infty/c][ \nu^\alpha X. A[\infty/c]/X] \vdash B[\infty/c]}
\infer1{\cns[\infty/c](\alpha < \beta), \Gamma[\infty/c], \nu^\beta X. A[\infty/c] \vdash B[\infty/c]}
\end{prooftree}}
\) 
\\[2.0em]
\( \left(
\scalebox{0.7}{\begin{prooftree}
\hypo{u_0}
\infer1{\cns, \Gamma \vdash  A[\mu^a X.A/X] }
\infer1{\cns, \Gamma \vdash \mu^\infty X. A}
\end{prooftree}}
\right)[\orda := \infty] \) & \(=\) & \(
\scalebox{0.7}{\begin{prooftree}
\hypo{u_0[a := \infty]}
\infer1{\cns[\infty/a], \Gamma[\infty/a], A[\infty/a][ \mu^\infty X. A[\infty/a]/X] }
\infer1{\cns[\infty/a], \Gamma[\infty/a] \vdash \mu^\infty X. A[\infty/a]}
\end{prooftree}}
\) 
\\[2.0em]
\( \left(
\scalebox{0.7}{\begin{prooftree}
\hypo{u_0}
\infer1{\cns(\alpha < \beta), \Gamma \vdash A[\mu^\alpha X.A/X] }
\infer1{\cns(\alpha < \beta), \Gamma \vdash \mu^\beta X. A}
\end{prooftree}}
\right)[\ordc := \infty] \) & \(=\) & \(
\scalebox{0.7}{\begin{prooftree}
\hypo{u_0[c := \infty]}
\infer1{\cns[\infty/c](\alpha < \beta), \Gamma[\infty/c] \vdash A[\infty/c][ \mu^\alpha X. A[\infty/c]/X]}
\infer1{\cns[\infty/c](\alpha < \beta), \Gamma[\infty/c] \vdash \mu^\beta X. A[\infty/c]}
\end{prooftree}}
\) 
\\[2.0em]
\( \left(
\scalebox{0.7}{\begin{prooftree}
\hypo{u_0}
\infer1{\cns+_a b, \Gamma, A[\mu^b X. A/X] \vdash B}
\infer1{\cns, \Gamma, \mu^a X. A \vdash B}
\end{prooftree}}
\right)[\orda := \infty] \) & \(=\) & \(
\scalebox{0.7}{\begin{prooftree}
\hypo{u_0[a := \infty]}
\infer1{\cns[\infty/a]+_\infty b, \Gamma[\infty/a], A[\infty/a][\mu^b X. A[\infty/a]/X] \vdash B[\infty/a]}
\infer1{\cns[\infty/a], \Gamma[\infty/a], \mu^\infty X. A[\infty/a] \vdash B[\infty/a]}
\end{prooftree}}
\) 
\\[2.0em]
\( \left(
\scalebox{0.7}{\begin{prooftree}
\hypo{u_0}
\infer1{\cns+_\alpha b, \Gamma, A[\mu^b X. A/X] \vdash B}
\infer1{\cns, \Gamma, \mu^\alpha X. A \vdash B}
\end{prooftree}}
\right)[\ordc := \infty] \) & \(=\) & \(
\scalebox{0.7}{\begin{prooftree}
\hypo{u_0[c := \infty]}
\infer1{\cns[\infty/c]+_\alpha b, \Gamma[\infty/c], A[\infty/c][\mu^b X. A[\infty/c]/X] \vdash B[\infty/c]}
\infer1{\cns[\infty/c], \Gamma[\infty/c], \mu^\alpha X. A[\infty/c] \vdash B[\infty/c]}
\end{prooftree}}
\) 
\end{tabular}
\]
\caption{\label{f:substitution} Substitution for ordinal variables}
\end{figure}

\begin{prop}
If $u$ is a valid proof of the sequent $\cns, \Gamma \vdash A$ and $a$ is a child of $\infty$ in $\cns$, then $u[a:=\infty]$ is a valid proof of $\cns[\infty/a], \Gamma[\infty/a] \vdash A[\infty/a]$.
\end{prop}
\begin{proof}
A routine check.
\end{proof}

It is clear that the operation of substituting $\infty$ for children of $\infty$ in the end sequent of a proof $u$ can be repeated until the constraint of the end sequent becomes trivial. We denote the (unique) proof obtained in this way by $u_\infty$.

\subsection{Reduction rules}

We are now ready to provide the rewrite rules - cut reductions and permutation rules -  that will provide computational content to proofs. We will present these rules using both the term notation for proofs and in the sequent calculus style; in the latter presentation we will annotate formulas on the left with variables to clarify the connection with the term presentation. The rewrite rules are not sufficient for a full cut elimination theorem, which is not our goal; we only include those rules that are necessary for the computational interpretation of proofs that we give in the next section. For instance, there are no permutation rules for left logical rules. 

We begin with reduction rules. Below, given a multi-set of formulas $\Gamma$, we abbreviate a sequent $\vec{w}$ consisting of one proof of the sequent $\cns \vdash A$ for each formula $A$ in $\Gamma$ by:
\[
\scalebox{0.7}{\begin{prooftree}
\hypo{\vec{w}}
\infer1{\cns \vdash \Gamma}
\end{prooftree}}
\]

\paragraph{$\vee$-reduction}

\[\mcut(\vec{w},\vee_R^i(v), \vec{z},x,\vee_L(x,y_0,y_1,u_0,u_1)) \longrightarrow \mcut(\vec{w},v,\vec{z},y_i,u_i)\]

\[
\scalebox{0.7}{\begin{prooftree}
\hypo{\vec{w}}
\infer1{\cns \vdash \Gamma}
\hypo{v}
\infer1{\cns \vdash A_i}
\infer1[$\vee_R$]{\cns \vdash A_0 \vee A_1}
\hypo{u_0}
\infer1{\cns, \vec{z} : \Gamma, y_0 : A_0 \vdash B}
\hypo{u_1}
\infer1{\cns,\vec{z} : \Gamma, y_ 1 : A_1 \vdash B}
\infer2[$\vee_L$]{\cns, \vec{z} : \Gamma, x : A_0 \vee A_1 \vdash B}
\infer3[$\mcut$]{\cns \vdash B}
\end{prooftree}}
\]
\[\Downarrow\]
\[
\scalebox{0.7}{\begin{prooftree}
\hypo{\vec{w}}
\infer1{\cns \vdash \Gamma}
\hypo{v}
\infer1{\cns\vdash A_i}
\hypo{u_i}
\infer1{\cns,\vec{z} : \Gamma, y_i : A_i \vdash B}
\infer3[$\mcut$]{\cns \vdash B}
\end{prooftree}}
\]

\paragraph{$\wedge$-reduction}

\[\mcut(\vec{u},\lprod_R(v_0,v_1), \vec{x},y,\lprod_L(y,z_0,z_1,w)) \longrightarrow \mcut(\vec{u},v_0,v_1,\vec{x},z_0,z_1,w)\]

\[
\scalebox{0.7}{\begin{prooftree}
\hypo{\vec{u}}
\infer1{\cns \vdash \Gamma}
\hypo{v_0}
\infer1{\cns \vdash  A_0}
\hypo{v_1}
\infer1{\cns\vdash A_1}
\infer2[$\wedge_R$]{\cns \vdash A_0 \wedge A_1}
\hypo{w}
\infer1{\cns,\vec{x} : \Gamma, z_0 : A, z_1 : B\vdash C}
\infer1[$\wedge_L$]{\cns,\vec{x} : \Gamma, y : A_0 \wedge A_1 \vdash C}
\infer3[$\mcut$]{\cns\vdash C}
\end{prooftree}}
\]
\[\Downarrow\]
\[
\scalebox{0.7}{\begin{prooftree}
\hypo{\vec{u}}
\infer1{\cns \vdash \Gamma}
\hypo{v_0}
\infer1{\cns\vdash  A_0}
\hypo{v_1}
\infer1{\cns \vdash A_1}
\hypo{w}
\infer1{\cns, \vec{x} : \Gamma, z_0 : A, z_1 : B\vdash C}
\infer4[$\mcut$]{\cns \vdash C}
\end{prooftree}}
\]

\paragraph{$\lfun$-reduction}

\begin{align*}
& \mathrel{\phantom{\longrightarrow}} \mcut(\vec{u}_0,\vec{u}_1, \lfun_R\!(x,v)),\vec{y_0},\vec{y}_1, z_0,\lfun_L\!(z_0,z_1,w_0,w_1)) \\
& \longrightarrow \mcut(\vec{u}_1,\mcut( \mcut(\vec{u}_0, \vec{y}_0,w_0),x, v), \vec{y}_1, z_1,  w_1)
\end{align*}

\[
\scalebox{0.7}{\begin{prooftree}
\hypo{\vec{u}_0}
\infer1{\cns \vdash \Gamma_0}
\hypo{\vec{u}_1}
\infer1{\cns \vdash \Gamma_1}
\hypo{v}
\infer1{\cns, x : A \vdash B}
\infer1[$\lfun_R$]{\cns \vdash A \lfun B}
\hypo{w_0}
\infer1{\cns, \vec{y}_0 : \Gamma_0 \vdash A}
\hypo{w_1}
\infer1{\cns, \vec{y}_1 : \Gamma_1, z_1 : B \vdash C}
\infer2[$\lfun_L$]{\cns, \vec{y}_0,\vec{y}_1 : \Gamma_0, \Gamma_1, z_0 : A \to B \vdash C}
\infer4[$\mcut$]{\cns\vdash C}
\end{prooftree}}
\]
\[\Downarrow\]
\[
\scalebox{0.7}{\begin{prooftree}
\hypo{\vec{u}_1}
\infer1{\cns \vdash \Gamma_1}
\hypo{\vec{u}_0}
\infer1{\cns \vdash  \Gamma_0}
\hypo{w_0}
\infer1{\cns, \vec{y}_0 : \Gamma_0 \vdash A}
\infer2[$\mcut$]{\cns \vdash A}
\hypo{v}
\infer1{\cns, x : A \vdash B}
\infer2[$\mcut$]{\cns \vdash B}
\hypo{w_1}
\infer1{\cns, \vec{y}_1 : \Gamma_1, z_1 : B \vdash C}
\infer3[$\mcut$]{\cns \vdash C}
\end{prooftree}}
\]

\paragraph{$\mu$-reduction}
\[\mcut(\vec{u}, \mu_R(\infty,\infty, v),\vec{x},y,\mu_L(\infty,a,y,z,w) ) \longrightarrow \mcut(\vec{u}, v,\vec{x},z,w[a:=\infty])\]

\[
\scalebox{0.7}{\begin{prooftree}
\hypo{\vec{u}}
\infer1{\cns \vdash \Gamma}
\hypo{v}
\infer1{\cns \vdash A[\mu^\infty X.A/X]}
\infer1{\cns \vdash \mu^\infty X. A}
\hypo{w}
\infer1{\cns+_\infty a, \vec{x} : \Gamma, z : A[\mu^a X.A/X] \vdash B}
\infer1{\cns, \vec{x} : \Gamma, y : \mu^\infty X.A \vdash B}
\infer3{\vdash B}
\end{prooftree}}
\]
\[\Downarrow\]
\[
\scalebox{0.7}{\begin{prooftree}
\hypo{\vec{u}}
\infer1{\cns \vdash \Gamma}
\hypo{v}
\infer1{\cns \vdash A[\mu^\infty X.A/X]}
\hypo{w[a := \infty]}
\infer1{\cns, \vec{x} : \Gamma, z : A[\mu^\infty X.A/X] \vdash B}
\infer3{\cns \vdash B}
\end{prooftree}}
\]

\paragraph{$\nu$-reduction}

\[\mcut(\vec{u}, \nu_R(\infty,a, v),\vec{x},y,\nu_L(\infty,\infty,y,z,w) ) \longrightarrow \mcut(\vec{u}, v[a:=\infty],\vec{x},z,w)\]

\[
\scalebox{0.7}{\begin{prooftree}
\hypo{\vec{u}}
\infer1{\cns, \Gamma}
\hypo{v}
\infer1{\cns+_\infty a, \vdash A[\nu^a X.A/X]}
\infer1{\cns \vdash \nu^\infty X. A}
\hypo{w}
\infer1{\cns,\vec{x} : \Gamma, z : A[\nu^\infty X.A/X] \vdash B}
\infer1{\cns, \vec{x} : \Gamma, y : \nu^\infty X.A \vdash B}
\infer3{\cns \vdash B}
\end{prooftree}}
\]
\[\Downarrow\]
\[
\scalebox{0.7}{\begin{prooftree}
\hypo{\vec{u}}
\infer1{\cns \vdash \Gamma}
\hypo{v[a:=\infty]}
\infer1{\cns \vdash A[\nu^\infty X.A/X]}
\hypo{w}
\infer1{\cns, \vec{x} : \Gamma, z : A[\nu^\infty X.A/X] \vdash B}
\infer3{\cns \vdash B}
\end{prooftree}}
\]

\paragraph{$\exists$-reduction}
\[\mcut(\vec{u}, \exists_R(\infty,\infty, v),\vec{x},y,\exists_L(\infty,a,y,z,w) ) \longrightarrow \mcut(\vec{u}, v,\vec{x},z,w[a:=\infty])\]

\[
\scalebox{0.7}{\begin{prooftree}
\hypo{\vec{u}}
\infer1{\cns \vdash \Gamma}
\hypo{v}
\infer1{\cns \vdash A[\infty/c]}
\infer1{\cns \vdash \exists c < \infty. A}
\hypo{w}
\infer1{\cns+_\infty a, \vec{x} : \Gamma, z : A[a/c] \vdash B}
\infer1{\cns, \vec{x} : \Gamma, y : \exists c < \infty : A \vdash B}
\infer3{\vdash B}
\end{prooftree}}
\]
\[\Downarrow\]
\[
\scalebox{0.7}{\begin{prooftree}
\hypo{\vec{u}}
\infer1{\cns \vdash \Gamma}
\hypo{v}
\infer1{\cns \vdash A[\infty/c]}
\hypo{w[a := \infty]}
\infer1{\cns, \vec{x} : \Gamma, z : A[\infty/c] \vdash B}
\infer3{\cns \vdash B}
\end{prooftree}}
\]

\paragraph{$\forall$-reduction}

\[\mcut(\vec{u}, \forall_R(\infty,a, v),\vec{x},y,\forall_L(\infty,\infty,y,z,w) ) \longrightarrow \mcut(\vec{u}, v[a:=\infty],\vec{x},z,w)\]

\[
\scalebox{0.7}{\begin{prooftree}
\hypo{\vec{u}}
\infer1{\cns, \Gamma}
\hypo{v}
\infer1{\cns+_\infty a \vdash A[a/c]}
\infer1{\cns \vdash \forall c < \infty: A}
\hypo{w}
\infer1{\cns,\vec{x} : \Gamma, z : A[\infty/c] \vdash B}
\infer1{\cns, \vec{x} : \Gamma, y : \forall c < \infty : A \vdash B}
\infer3{\cns \vdash B}
\end{prooftree}}
\]
\[\Downarrow\]
\[
\scalebox{0.7}{\begin{prooftree}
\hypo{\vec{u}}
\infer1{\cns \vdash \Gamma}
\hypo{v[a:=\infty]}
\infer1{\cns \vdash A[\infty/c]}
\hypo{w}
\infer1{\cns, \vec{x} : \Gamma, z : A[\infty/c] \vdash B}
\infer3{\cns \vdash B}
\end{prooftree}}
\]

\paragraph{Identity reduction}

\[\mcut(u,x,\mathsf{id}(x)) \longrightarrow u\]

\[
\scalebox{0.7}{\begin{prooftree}
\hypo{u}
\infer1{\cns \vdash A}
\hypo{}
\infer1[$\mathsf{id}$]{\cns, x : A \vdash A}
\infer2[$\mcut$]{\cns \vdash A}
\end{prooftree}}
\quad
\Rightarrow
\quad
\scalebox{0.7}{\begin{prooftree}
\hypo{u}
\infer1{\cns \vdash A}
\end{prooftree}}
\]

\paragraph{Reduction of contraction}

\[\mcut(\vec{u},v,\vec{x},y,C(y,z_1,z_2,w)) \longrightarrow \mcut(\vec{u},v,v,\vec{x},z_1,z_2,w)\]

\[
\scalebox{0.7}{\begin{prooftree}
\hypo{\vec{u}}
\infer1{\cns \vdash \Gamma}
\hypo{v}
\infer1{A}
\hypo{w}
\infer1{\cns, \vec{x} : \Gamma, z_1 : A, z_2 : A \vdash B}
\infer1{\cns, \vec{x} : \Gamma, y : A \vdash B}
\infer3[$\mcut$]{\cns \vdash B}
\end{prooftree}}
\quad
\Rightarrow
\quad
\scalebox{0.7}{\begin{prooftree}
\hypo{\vec{u}}
\infer1{\cns \vdash \Gamma}
\hypo{v}
\infer1{A}
\hypo{v}
\infer1{A}
\hypo{w}
\infer1{\cns, \vec{x} : \Gamma, z_1 : A, z_2 : A \vdash B}
\infer4[$\mcut$]{\cns \vdash B}
\end{prooftree}}
\]

\paragraph{Reduction of weakening}

\[\mcut(\vec{u},v,\vec{x},y,W(y,w)) \longrightarrow \mcut(\vec{u},\vec{x},w)\]

\[
\scalebox{0.7}{\begin{prooftree}
\hypo{\vec{u}}
\infer1{\cns \vdash \Gamma}
\hypo{v}
\infer1{A}
\hypo{w}
\infer1{\cns, \vec{x} : \Gamma \vdash B}
\infer1{\cns, \vec{x} : \Gamma, y : A \vdash B}
\infer3[$\mcut$]{\cns \vdash B}
\end{prooftree}}
\quad
\Rightarrow
\quad
\scalebox{0.7}{\begin{prooftree}
\hypo{\vec{u}}
\infer1{\cns \vdash \Gamma}
\hypo{w}
\infer1{\cns, \vec{x} : \Gamma  \vdash B}
\infer2[$\mcut$]{\cns \vdash B}
\end{prooftree}}
\]

Next we introduce permutation rules:

\paragraph{Right $\vee$-permutation}
\[\mcut(\vec{u},\vec{x},\vee_R^i(v)) \longrightarrow \vee_R^i(\mcut(\vec{u},\vec{x},v)) \]

\[
\scalebox{0.7}{\begin{prooftree}
\hypo{\vec{u}}
\infer1{\cns \vdash \Gamma}
\hypo{v}
\infer1{\cns,\Gamma \vdash A_i}
\infer1{\cns, \Gamma \vdash A_0 \vee A_1}
\infer2[$\mcut$]{\cns \vdash A_0 \vee A_1}
\end{prooftree}}
\quad
\Rightarrow
\quad
\scalebox{0.7}{\begin{prooftree}
\hypo{\vec{u}}
\infer1{\cns \vdash \Gamma}
\hypo{v}
\infer1{\cns,\Gamma \vdash A_i}
\infer2[$\mcut$]{\cns \vdash A_i}
\infer1{\cns \vdash A_0 \vee A_1}
\end{prooftree}}
\]

\paragraph{Right $\lprod$-permutation}
\[\mcut(\vec{u}_1,\vec{u_2},\vec{x}_1,\vec{x}_2,\wedge_R(v,w)) \longrightarrow \wedge_R(\mcut(\vec{u}_1, \vec{x}_1,v), \mcut(\vec{u}_2,\vec{x}_2,w))\]

\[
\scalebox{0.7}{\begin{prooftree}
\hypo{\vec{u}_1}
\infer1{\cns \vdash \Gamma}
\hypo{\vec{u}_2}
\infer1{\cns \vdash \Delta}
\hypo{v}
\infer1{\cns,\Gamma \vdash A}
\hypo{w}
\infer1{\cns, \Delta \vdash B}
\infer2{\cns, \Gamma,\Delta \vdash A \lprod B}
\infer3[$\mcut$]{\cns \vdash A \lprod B}
\end{prooftree}}
\quad
\Rightarrow
\quad
\scalebox{0.7}{\begin{prooftree}
\hypo{\vec{u}_1}
\infer1{\cns \vdash \Gamma}
\hypo{v}
\infer1{\cns,\Gamma \vdash A}
\infer2[$\mcut$]{\cns \vdash A}
\hypo{\vec{u}_2}
\infer1{\cns \vdash \Delta}
\hypo{w}
\infer1{\cns, \Delta \vdash B}
\infer2[$\mcut$]{\cns \vdash B}
\infer2{\cns \vdash A \lprod B}
\end{prooftree}}
\]
\paragraph{Right $\lfun$-permutation}
\[\mcut(\vec{u},\vec{x}, {\lfun_R}(z,v)) \longrightarrow {\lfun_R}(z,\mcut(\vec{x},\vec{u},v))\]
\[
\scalebox{0.7}{\begin{prooftree}
\hypo{\vec{u}}
\infer1{\cns \vdash \Gamma}
\hypo{v}
\infer1{\cns,\Gamma, A \vdash B}
\infer1{\cns, \Gamma \vdash A \to B}
\infer2[$\mcut$]{\cns \vdash A \to B}
\end{prooftree}}
\quad
\Rightarrow
\quad
\scalebox{0.7}{\begin{prooftree}
\hypo{\vec{u}}
\infer1{\cns \vdash \Gamma}
\hypo{v}
\infer1{\cns,\Gamma, A \vdash B}
\infer2[$\mcut$]{\cns, A \vdash B}
\infer1{\cns \vdash A \to B}
\end{prooftree}}
\]

\paragraph{Right $\mu$-permutation}

\[\mcut(\vec{u},\vec{x}, {\mu_R}(\infty,\infty, v)) \longrightarrow {\mu_R}(\infty,\infty,\mcut(\vec{x},\vec{u},v))\]
\[
\scalebox{0.7}{\begin{prooftree}
\hypo{\vec{u}}
\infer1{\cns \vdash \Gamma}
\hypo{v}
\infer1{\cns,\Gamma  \vdash A[\mu^\infty X. A/X]}
\infer1{\cns, \Gamma \vdash \mu^\infty X. A}
\infer2[$\mcut$]{\cns \vdash \mu^\infty X. A}
\end{prooftree}}
\quad
\Rightarrow
\quad
\scalebox{0.7}{\begin{prooftree}
\hypo{\vec{u}}
\infer1{\cns \vdash \Gamma}
\hypo{v}
\infer1{\cns,\Gamma \vdash A[\mu^\infty X.A/X]}
\infer2[$\mcut$]{\cns \vdash A[\mu^\infty X.A/X]}
\infer1{\cns \vdash \mu^\infty X. A}
\end{prooftree}}
\]

\paragraph{Right $\nu$-permutation}
\[\mcut(\vec{u},\vec{x}, {\nu_R}(\infty,a, v)) \longrightarrow {\nu_R}(\infty,a,\mcut(\vec{x},\vec{u}',v))\]

\[
\scalebox{0.7}{\begin{prooftree}
\hypo{\vec{u}}
\infer1{\cns \vdash \Gamma}
\hypo{v}
\infer1{\cns+_\infty a,\Gamma  \vdash A[\nu^a X. A/X]}
\infer1{\cns, \Gamma \vdash \nu^\infty X. A}
\infer2[$\mcut$]{\cns \vdash \nu^\infty X. A}
\end{prooftree}}
\quad
\Rightarrow
\quad
\scalebox{0.7}{\begin{prooftree}
\hypo{\vec{u}'}
\infer1{\cns+_\infty a \vdash \Gamma}
\hypo{v}
\infer1{\cns+_\infty a,\Gamma \vdash A[\nu^a X.A/X]}
\infer2[$\mcut$]{\cns +_\infty a \vdash A[\nu^a X.A/X]}
\infer1{\cns \vdash \nu^\infty X. A}
\end{prooftree}}
\]

\paragraph{Right $\exists$-permutation}
\[\mcut(\vec{u},\vec{x}, {\exists_R}(\infty,\infty, v)) \longrightarrow {\exists_R}(\infty,\infty,\mcut(\vec{x},\vec{u},v))\]

\[
\scalebox{0.7}{\begin{prooftree}
\hypo{\vec{u}}
\infer1{\cns \vdash \Gamma}
\hypo{v}
\infer1{\cns,\Gamma  \vdash A[\infty/c]}
\infer1{\cns, \Gamma \vdash \exists c < \infty: A}
\infer2[$\mcut$]{\cns \vdash \exists c < \infty: A}
\end{prooftree}}
\quad
\Rightarrow
\quad
\scalebox{0.7}{\begin{prooftree}
\hypo{\vec{u}}
\infer1{\cns \vdash \Gamma}
\hypo{v}
\infer1{\cns,\Gamma \vdash A[\infty/c]}
\infer2[$\mcut$]{\cns \vdash A[\infty/c]}
\infer1{\cns \vdash \exists c < \infty: A}
\end{prooftree}}
\]

\paragraph{Right $\forall$-permutation}

\[\mcut(\vec{u},\vec{x}, {\forall_R}(\infty,a, v)) \longrightarrow {\forall_R}(\infty,a,\mcut(\vec{x},\vec{u},v))\]
\[
\scalebox{0.7}{\begin{prooftree}
\hypo{\vec{u}}
\infer1{\cns \vdash \Gamma}
\hypo{v}
\infer1{\cns+_\infty a,\Gamma  \vdash A[a/c]}
\infer1{\cns, \Gamma \vdash \forall c < \infty: A}
\infer2[$\mcut$]{\cns \vdash \forall c < \infty: A}
\end{prooftree}}
\quad
\Rightarrow
\quad
\scalebox{0.7}{\begin{prooftree}
\hypo{\vec{u}'}
\infer1{\cns+_\infty a \vdash \Gamma}
\hypo{v}
\infer1{\cns+_\infty a,\Gamma \vdash A[a/c]}
\infer2[$\mcut$]{\cns +_\infty a \vdash A[a/c]}
\infer1{\cns \vdash \forall c < \infty: A}
\end{prooftree}}
\]

Here, the proofs $\vec{u}'$ are obtained from $\vec{u}$ by carrying along the eigenvariable $a$, where we may assume without loss of generality that $a$ does not occur in any constraint in any proof in $\vec{u}$.

\paragraph{Rewrite relation}

With the reduction and permutation rules in place, we now define the rewrite relation $u \longrightarrow v$ between proof terms. 

\begin{defi}
Given proofs $u,v$ we say that $u$ rewrites to $v$ in one step, written $u \longrightarrow^1 v$, if $v$ results from applying one of the rewrite rules to a subterm of $u$ subject to the restriction that rewrite rules may only be applied to lowest reducible cuts. That is, if $u$ contains a subterm of the form $\cut(\vec{w},\vec{x},u')$ to which one of the rewrite rules applies, then no rewrite rule may be applied to any subterm of $u'$  nor any subterm of any of the terms $\vec{w}$. 

We write $u \longrightarrow v$ for the reflexive transitive closure of the one-step rewrite relation $u \longrightarrow^1 v$.  
\end{defi}

With the restriction to lowest reducible cuts, it is easy to establish confluence:
\begin{prop}[Confluence]
\label{p:confluence}
If $u \longrightarrow v$ and $u \longrightarrow v'$ then there is a proof $w$ such that $v \longrightarrow w$ and $v' \longrightarrow w$.
\end{prop}
\begin{proof}
We introduce a relation of parallel reduction $\Longrightarrow$, show that it is confluent and that its reflexive transitive closure is equal to $\longrightarrow$. It follows from this that the rewrite relation $\longrightarrow$ is also confluent, since the reflexive transitive closure of a confluent relation is also confluent. 

The relation $\Longrightarrow$ is defined to be the smallest relation on proof trees closed under the following rules:
\begin{itemize}
\item $u \Longrightarrow u$.
\item If $u \longrightarrow^1 v$ then $u \Longrightarrow v$.
\item If $v_1,\hdots,v_n$ are disjoint subtrees of $u$, $v_i \Longrightarrow v_i'$ for each $i \in \{1,\hdots,n\}$, and $u'$ is the result of replacing each $v_i$ by $v_i'$ in $u$, then $u \Longrightarrow u'$.
\end{itemize}
We write $\Longrightarrow^*$ for the reflexive transitive closure of $\Longrightarrow$. Since $\longrightarrow$ is the reflexive transitive closure of $\longrightarrow^1$, the inclusion $\longrightarrow\; \subseteq \; \Longrightarrow^*$ is immediate. For the reverse inclusion, we show that $u \Longrightarrow v$ implies $u \longrightarrow v$. This is by a straightforward induction on the derivation of $u \Longrightarrow v$ according to the rules given above:
\begin{itemize}
\item If $u = v$ then $u \longrightarrow v$ since $\longrightarrow$ is reflexive by definition. 
\item If $u \longrightarrow^1 v$ then $u \longrightarrow v$ since $\longrightarrow$ is the reflexive transitive closure of $\longrightarrow^1$.
\item If $v_1,\hdots,v_n$ are disjoint subtrees of $u$, $v_i \Longrightarrow v_i'$ for each $i \in \{1,\hdots,n\}$, and $u'$ is the result of replacing each $v_i$ by $v_i'$ in $u$, then we can reason by sub-induction on the number of subtrees $n$:
\begin{itemize}
\item For $n = 1$, $u'$ is the result of replacing a subtree $v$ of $u$ by some $v'$ such that $v \Longrightarrow v'$. By the main induction hypothesis, $v \longrightarrow v'$, and so $u \longrightarrow u'$. 
\item Supposing $n = m + 1$, where the sub-induction hypothesis holds for $m$, let $u''$ be the result of replacing $v_1$ by $v_1'$ in $u$. The main induction hypothesis  gives $v_1 \longrightarrow v_1'$, hence $u \longrightarrow u''$. The sub-induction hypothesis on $m$ yields $u'' \longrightarrow u'$. Since $\longrightarrow$ is transitive, we get $u \longrightarrow u'$ as required. 
\end{itemize}
\end{itemize}

It remains to show that $\Longrightarrow$ is confluent.  First, if  $u \Longrightarrow v$, then a straightforward induction on the derivation of $u \Longrightarrow v$ shows that there must be a finite set $\{u_1,\hdots u_n\}$ of subtrees of $u$, such that:
\begin{itemize}
\item each subtree $u_i$ ends with a multicut,
\item $v$ is the result of replacing each $u_i$ by some $v_i$, where $v_i$ is obtained by applying a reduction or permutation rule to $u_i$, and
\item the subtrees $u_1,\hdots, u_n$ are pairwise disjoint. 
\end{itemize}
With this in mind, suppose that $u \Longrightarrow w$ and $u \Longrightarrow w'$. Let $\{u_1,\hdots,u_n\}$ and $\{u'_1,\hdots,u'_m\}$ be finite sets of subtrees of $u$, each ending with a multicut, such that:
\begin{itemize}
\item $w$ is the result of replacing each $u_i$ by some $v_i$, where $v_i$ is obtained by applying a reduction or permutation rule to $u_i$,
\item $w'$ is the result of replacing each $u'_j$ by some $v'_j$, where $v'_j$ is obtained by applying a reduction or permutation rule to $u'_j$,
\item the subtrees $u_1,\hdots, u_n$ are pairwise disjoint, and
\item the subtrees $u'_1,\hdots,u'_m$ are pairwise disjoint.
\end{itemize}
Now, let $i \in \{1,\hdots, n\}$ and $j \in \{1,\hdots, m\}$ be such that the subtrees $u_i$ and $u'_j$ are \emph{not} disjoint, meaning that one is a subtree of the other. Since reductions and permutations may only be applied to lowest reducible cuts in $u$, this means that in fact $u_i = u'_j$. If the rightmost subtree of $u_i$ ends with a right logical rule, then $v_i$ must result from applying the corresponding right permutation rule to $u_i$, and $v'_j$ must result from appying the same right permutation rule to $u'_j$. Similarly, if the rightmost subtree of $u_i$ ends with a left logical rule, a weakening or a contraction, then $v_i$ results from applying the corresponding reduction rule to $u_i$, and likewise $v'_j$ results from applying the same reduction rule to $u'_j$. So let $w''$ be the result of replacing each subtree $u_i$ in $u$ by $v_i$, and each subtree $u'_j$ of $u$ by $v'_j$; we have just shown that no conflict arises in this definition of $w''$. Furthermore, we have $w \Longrightarrow w''$ and $w' \Longrightarrow w''$ as required. 
\end{proof}
It should be emphasized that Proposition \ref{p:confluence} - which may seem surprising in the current setting of sequent calculus, where cut elimination is usually only confluent up to reordering of inference rules - crucially relies on the restricted set of reduction and permutation rules that we have included here. These rules have been chosen precisely in order for the computability semantics that we will set up in the next section to work as smoothly as possible. If our goal here were to obtain a full cut elimination result for the logic, additional rules would need to be included, in particular left permutation rules. And then, Proposition \ref{p:confluence} would no longer hold. For example, consider the following proof:
\[
\scalebox{0.7}{
\begin{prooftree}
\hypo{}
\infer1[$id$]{A \vdash A}
\hypo{}
\infer1[$id$]{A \vdash A}
\infer2[$\vee_L$]{A \vee A \vdash A}
\hypo{}
\infer1[$id$]{A \vdash A }
\infer1[$\vee_R$]{A \vdash A \vee A}
\infer2[$cut$]{A \vee A \vdash A \vee A}
\end{prooftree}
}
\]
Here, with a  left permutation rule in place, we would have a choice of two different ways of reducing the proof, according to whether a permutation of the right $\vee$-rule or the left $\vee$-rule is performed first. These two reduction strategies would yield the following two normal forms:
\[
\scalebox{0.7}{
\begin{prooftree}
\hypo{}
\infer1[$id$]{A \vdash A}
\hypo{}
\infer1[$id$]{A \vdash A}
\infer2[$\vee_L$]{A \vee A \vdash A}
\infer1[$\vee_R$]{A \vee A \vdash A \vee A}
\end{prooftree}
\qquad
\begin{prooftree}
\hypo{}
\infer1[$id$]{A \vdash A}
\infer1[$\vee_R$]{A \vdash A \vee A}
\hypo{}
\infer1[$id$]{A \vdash A}
\infer1[$\vee_R$]{A \vdash A \vee A}
\infer2[$\vee_L$]{A \vee A \vdash A \vee A}
\end{prooftree}
}
\]

\section{Computability}

In this section we will provide our computational interpretation of our proof system, by defining a notion of \emph{computability} for closed derivations. Our goal is to show that every valid proof is computable in this sense.

\subsection{Definition of computability}

To define the notion of computability, we will make use of an extended, infinitary version of the language in which fixpoints may be annotated by actual ordinals. The syntax is as for $\mathcal{L}_{\mathsf{ID}}$, except that we extend ordinal terms to include actual ordinals. 

Let $D$ be the set of all derivations. Throughout the rest of the paper, we let $\Omega$ denote a fixed ordinal with cardinality greater than that of $D$ (i.e.  greater than $2^{\aleph_0}$).
We now define the syntax of the infinitary language by the following grammar:
\[\mathcal{I}_{\mathsf{ID}} \ni A:= \top \mid A \lsum A \mid A \lprod A \mid A \lfun A \mid \mu^\alpha X.B \mid \nu^\alpha X. B \mid \exists a < \alpha : A \mid \forall a < \alpha: A\]
\[\mathcal{I}^+(X) \ni B := A \mid X \mid \top \mid B \lsum B \mid B \lprod B \mid A \to B\]
The only difference with the grammar for $\mathcal{L}_{\mathsf{ID}}$ is that, here, the metavariable $\alpha$ ranges over ordinal variables, the constant $\infty$, and actual ordinals $ \xi \leq \Omega +1$. Given an $\mathcal{I}_\mathsf{ID}$-formula $A$, its \emph{underlying $\mathcal{L}_{\mathsf{ID}}$-formula} $\underline{A}$ is defined by replacing all  occurrences of ordinals by $\infty$. For example:
\[\underline{\forall a < \omega: (\mu^a X.X \to \mu^{\omega + 1} X. X} ) = \forall a < \infty: (\mu^a X.X \to \mu^\infty X.X)\]
Note that if $A$ is a closed formula, then so is $\underline{A}$.

We can define an ordering $\prec$ over formulas in $\mathcal{I}_{\mathsf{ID}}$ as the smallest relation satisfying the following conditions:
\begin{itemize}
\item $A,B \prec A \circ B$ for $\circ \in \{\wedge, \vee, \lfun\}$,
\item for all ordinals $\xi$, $A[\eta^\xi X. A/X] \prec \eta^\infty X. A$ for $\eta \in \{\mu,\nu\}$,
\item for all ordinals $\rho < \xi$, $A[\eta^\rho X. A/X] \prec \eta^\xi X. A$ for $\eta \in \{\mu,\nu\}$.
\item for all ordinals $\xi$, $A[\xi/a] \prec \exists a < \infty : A$.
\item for all ordinals $\xi$, $A[\xi/a] \prec \forall a < \infty : A$.
\item for all ordinals $\rho < \xi$, $A[\rho/a] \prec \exists a < \xi : A$.
\item for all ordinals $\rho < \xi$, $A[\rho/a] \prec \forall a < \xi : A$.
\end{itemize}

\begin{prop}
The relation $\prec$ is wellfounded.
\end{prop}

\begin{proof}
For a contradiction, suppose that $A_1 A_2 A_3\hdots$ is an infinite series of  $\mathcal{I}_{\mathsf{ID}}$-formulas such that $A_{i+1} \prec A_i$ for all $i \in \mathbb{N}$. We claim that there must be some $k \in \mathbb{N}$ and some fixed formula $B \in \mathcal{I}_{\mathsf{ID}}^+(X)$ such that, for all $i \geq k$, $A_i$ is of the form:
\[C[\eta^{\xi} X. B/X]\]
for some ordinal $\xi$ and some subformula $C$ of $B$ that has at least one free occurrence of the fixpoint variable $X$. This is simply because the only clauses in the definition of $\prec$ that do not decrease the length of the formula are those of for fixpoint formulas, and we have no dependencies between fixpoint variables.

So fix such an $n$, and write $A_i = C_i[\eta^{\xi_i} X. B/X]$ for each $i \geq n$. We claim that $\xi_{i+1} \leq \xi_i$ whenever $n \leq i$, and $\xi_{i + 1} < \xi_i$ for infinitely many $i \geq n$. This leads to an infinite descending chain of ordinals, and hence a contradiction.  

To prove the claim we reason by a case distinction on the shape of the subformula $C_i$:
\paragraph{Case $C_i =X$:} then $A_i$ is $\eta^{\xi_i} X. B$, and $A_{i+1}$ must be $B[\eta^{\xi_{i+1} X. B}/X]$ where $\xi_{i + 1} < \xi_i$.

\paragraph{Case $C_i \in \mathcal{I}_{\mathsf{ID}}$:} This case cannot occur since $C_i$ must have at least one free occurrence of $X$. 

\paragraph{Case $C_i = D_0 \circ D_1$, $\circ \in \{\wedge, \vee, \lfun\}$:} then $A_i$ is $D_0[\eta^{\xi_i} X. B/] \circ D_1[\eta^{\xi_i} X.B/X]$ and $A_{i+1}$ is either $D_0[\eta^{\xi_i} X. B/]$ or $D_1[\eta^{\xi_i} X.B/X]$. In either case we have $\xi_{i + 1} = \xi_i$.

It is easy to see that the first of these cases must occur infinitely many times, and the contradiction thus follows.
\end{proof}

From now on, when we say that a proof proceeds by ``induction on the complexity'' of formulas in $\mathcal{I}_\mathsf{ID}$, we actually mean induction on the wellfounded relation $\prec$.

\begin{defi}
Given a closed formula $A$ of the language  $\mathcal{I}_\mathsf{ID}$ we define a set  $\rset{A}$ of closed derivations of the closed sequent $\vdash \underline{A}$ by induction:
\begin{itemize}
\item  $u \in \rset{\top}$ iff $u \longrightarrow \mathsf{ax}_\top$
\item  $u \in \rset{A \lprod B}$ iff $ u \longrightarrow \lprod_R(v,w)$ for some $v \in \rset{A}$ and  $w \in \rset{B}$.
\item  $u \in \rset{A_0 \lsum A_1}$ iff $u \longrightarrow \lsum_R^i(v)$ for $i \in \{0,1\}$ and $v \in \rset{A_i}$.
\item $u \in \rset{A \lfun B}$ iff $u \longrightarrow {\lfun_R}(x,v)$ for some $x,v$ such that $\mcut(w,x,v) \in \rset{B}$ for all $w \in \rset{A}$.
\item $u \in \rset{\mu^\infty X. A}$ iff $u \in \rset{\mu^\xi X. A}$ for  some ordinal $\xi \leq \Omega$.
\item $u \in \rset{\nu^\infty X. A}$ iff $u \in \rset{\nu^\xi X. A}$ for every ordinal $\xi \leq \Omega$.
\item $u \in \rset{\ann{\xi} X. A}$ iff  $ u \longrightarrow \mu_R(\infty,\infty, v)$ for  some $\rho < \xi$ and some  $v \in \rset{A[\ann{\rho} X. A/X]}$.
\item $u \in \rset{\nu^\xi X. A}_V$ iff $u \longrightarrow \nu_{R}(\infty,a,v)$ for some ordinal variable $a$ such that  $v[a := \infty] \in \rset{A[\nann{\rho} X. A/X]}$ for every ordinal $\rho < \xi$.
\item $u \in \rset{\exists a < \infty: A}$ iff  $ u \longrightarrow \exists_R(\infty,\infty, v)$ for  some $\rho \leq \Omega$ and some  $v \in \rset{A[\rho/a]}$.
\item $u \in \rset{\forall a < \infty: A}$ iff $u \longrightarrow \forall_{R}(\infty,c,v)$ for some ordinal variable $c$ such that  $v[c := \infty] \in \rset{A[\rho/c]}$ for every ordinal $\rho \leq \Omega$.
\item $u \in \rset{\exists a < \xi: A}$ if  $ u \longrightarrow \exists_R(\infty,\infty, v)$ for  some $\rho < \xi$ and some  $v \in \rset{A[\rho/a]}$.
\item $u \in \rset{\forall a < \xi: A}$ if $u \longrightarrow \forall_{R}(\infty,c,v)$ for some ordinal variable $c$ such that  $v[c := \infty] \in \rset{A[\rho/c]}$ for every ordinal $\rho < \xi$.
\end{itemize}
If $A$ is a formula in $\mathcal{L}_\mathsf{ID}$, we call $\rset{A}$ the set of \emph{computable} (closed) derivations of $A$. 
\end{defi} 

We shall often write $\vec{u} \in \rset{\Gamma}$ if $\vec{u} = (u_1,\hdots,u_n)$, $\Gamma = (A_1,\hdots,A_n)$ and $u_i \in \rset{A_i}$ for each $i$.

The computability predicate introduced by the previous definition is intended to capture the computational meaning of formulas as types, i.e. what a program of that type is supposed to do. For example, the clause for implication says that proof of $A \to B$ should produce a function that expects a value of type $A$ as input, and for any such input should return a value of type $B$. The clauses for the fixpoint formulas say that a proof of a least fixpoint $\mu^\infty X. A$ should produce a value that approximates the least fixpoint up to \emph{some} ordinal value $\xi \leq \Omega$, while a proof of a greatest fixpoint $\nu^\infty X. A$ should produce a value that approximates the fixpoint up to \emph{any} given ordinal value $\xi \leq \Omega$.

We begin by establishing two basic facts about computable derivations:

\begin{prop}
\label{p:closure-of-candidates}
If $v \in \rset{A}$ and $u \longrightarrow v$ then $u \in \rset{A}$. 
\end{prop}

\begin{proof}
Direct consequence of the definition.
 \end{proof}

\begin{prop}
\label{p:large-ordinal}
Let $A $ be any formula in $\mathcal{I}_\mathsf{ID}^+(X)$. Then the following conditions hold:
\begin{description}
\item[Monotonicity] If $\rset{B} \subseteq \rset{C}$ then $\rset{A[B/X]}\subseteq \rset{A[C/X]}$. 
\item[Inflation] If $\rho \leq \xi$ then $\rset{\mu^\rho X. A} \subseteq \rset{\mu^\xi X. A}$.
\item[Deflation] If $\rho \leq \xi$ then $\rset{\nu^\rho X. A} \supseteq \rset{\nu^\xi X. A}$.
\item[Closure] $\rset{\eta^{\Omega + 1} X. A} = \rset{\eta^\Omega X. A}$, for $\eta \in \{\mu,\nu\}$.
\end{description}
\end{prop}
\begin{proof}
The conditions of Monotonicity, Inflation and Deflation are straightforward consequences of the definition of computability. For Closure, by the Inflation condition (if $\eta = \mu$) or the  Deflation condition (if $\eta = \nu$), the sets $\rset{\eta^\xi X. A}$ for $\xi < \Omega$ form a linearly ordered family of subsets of the set $D$ of derivations under the subsethood relation. Since a family of subsets of $D$ linearly ordered by subsethood can contain at most $\vert D \vert$ distinct subsets, it follows that there must be ordinals $\rho < \xi < \Omega$ such that $\rset{\eta^\rho X.A} = \rset{\eta^\xi X. A}$ (by the pigeonhole principle, since by our choice of  $\Omega$ there are more than  $\vert D \vert$ many ordinals $< \Omega$). We show that $\rset{\eta^\xi X.A} = \rset{\eta^{\xi'} X. A}$ for every ordinal $\xi' \geq \xi$, from which follows that $\rset{\eta^{\Omega + 1} X. A} = \rset{\eta^\Omega X. A}$ as required since $\xi < \Omega$. 

To prove the claim, consider the case of $\eta = \mu$. We prove by transfinite induction that, for all ordinals $\alpha \leq \Omega + 1$, we have: \[\xi \leq \alpha \Rightarrow \rset{\mu^{\xi} X. A} = \rset{\mu^{\alpha} X. A}\]
Let $\alpha$ be any such ordinal and assume that the statement holds for all $\beta < \alpha$. Suppose also that $\xi \leq \alpha$; we need to prove $\rset{\mu^{\xi} X. A} = \rset{\mu^{\alpha} X. A}$.  The inclusion $\rset{\mu^{\xi} X. A} \subseteq \rset{\mu^{\alpha} X. A}$ follows directly from Inflation since we assumed $\xi \leq \alpha$. For the converse direction, suppose $u \in \rset{\mu^{\alpha} X. A}$. Then  there is an ordinal $\beta < \alpha$ with $u \longrightarrow \mu_R(u_0,\infty,\infty)$ for some $u_0 \in \rset{A[\mu^\beta X. A/X]}$. If $\beta < \xi$ then it follows immediately from the definition of computability that $u \in \rset{\mu^\xi X. A}$. Otherwise $\beta \geq \xi$, and then it follows by the induction hypothesis on $\beta < \alpha$ that $\rset{\mu^\beta X.A} \subseteq \rset{\mu^\xi X. A} = \rset{\mu^\rho X. A}$ where we recall that $\rho < \xi$. Monotonicity gives:
\[\rset{A[\mu^\beta X. A/X]} \subseteq \rset{A[\mu^\rho X. A/X]}\]
so that $u_0 \in \rset{A[\mu^\rho X. A/X]}$. Since $\rho < \xi$ and $u \longrightarrow \mu_R(u_0,\infty,\infty)$ we again get $u \in \rset{\mu^\xi X. A}$ as required.

The proof for the case of $\eta = \nu$ is similar, and is omitted. 
\end{proof}

\subsection{Reflection of uncomputability}

Our proof of computability for valid closed proofs follows the same line of reasoning as in \cite{das2020circular}.
In this section we prove the main lemma used for that argument, which we will use to show that any uncomputable derivation has an infinite ``uncomputable branch''.  First, a technical definition.

\begin{defi}
An \emph{ordinal annotation}, or just annotation, is a map $f$ from some set $V$ of ordinal variables to ordinals. Given an annotation $f$ and a formula $A \in \mathcal{L}_\mathsf{ID}$ with $a \in \mathsf{dom}(f)$ whenever $a$ occurs in $A$, we define the $\mathcal{I}_\mathsf{ID}$-formula $A_f$ recursively as follows:
\begin{itemize}
\item $\top_f = \top$,
\item $(A \circ B)_f = A_f \circ B_f$ for $\circ \in \{\wedge, \vee, \lfun\}$,
\item $(\eta^\infty X. B)_f = \eta^\Omega X. B_f$ for $\eta \in \{\mu,\nu\}$,
\item $(\eta^a X. B)_f = \eta^{f(a)} X. B_f$ for $\eta \in \{\mu,\nu\}$,
\item $(\mathsf{Q} c < \infty: B)_f = \mathsf{Q} c < \infty: B_f$ for $\mathsf{Q} \in \{\exists ,\forall\}$,
\item $(\mathsf{Q} c < a: B)_f = \mathsf{Q} c < {f(a)}: B_f$ for $\mathsf{Q} \in \{\exists,\forall\}$.
\end{itemize}
\end{defi}
Given a multiset of formulas  $\Gamma$ and an annotation $f$ we write $\Gamma_f$ for the result of replacing each $A$ in $\Gamma$ by $A_f$.

\begin{defi}
A \emph{faithful annotation} of a sequent $\cns,\Gamma \vdash A$ is an annotation $f$ whose domain is the set of ordinal variables occurring in the constraint $\cns$, such that:
\begin{itemize}
\item $f(a) \leq \Omega$ for every ordinal variable $a$,
\item if $\cns \Vdash a < b$ then $f(a) < f(b)$.
\end{itemize}
\end{defi}

By convention we will write $f(\infty) = \Omega$, so that $f$ is defined for all ordinal terms (here we recall that ordinal terms are either ordinal variables or the constant $\infty$).
For the next proposition, we recall that $u_\infty$ is the derivation obtained from $u$ by repeatedly substituting $\infty$ for all ordinal variables in the constraint.

\begin{lem}
\label{p:one-step-descent}
For each derivation $u$ of a sequent $\cns, \Gamma \vdash A$, each faithful annotation $f$ of the sequent and tuple $\vec{v} \in \rset{\Gamma_f}$ of closed derivations,  if $\mcut(\vec{v},\vec{x}, u_\infty) \notin  \rset{A_f}$, then there exists a premiss $u^*$ of $u$ labelled $\cns^*,\Gamma^* \vdash A^*$, an extension $f^*$ of $f$ to $\cns^*$ and closed derivations $\vec{v}_*$ such that:
\begin{itemize}
\item $\vec{v}_* \in \rset{\Gamma^*_{f^*}}$,
\item $\mcut(\vec{v}_*,\vec{x}_*,u^*_{\infty}) \notin \rset{A^*_{f^*}}$.
\end{itemize}
\end{lem}
\begin{proof}
We prove Proposition \ref{p:one-step-descent} through a case-by-case analysis on the last rule applied in $u$. We leave the easy cases of the identity axiom and the axiom for $\top$ to the reader. We leave the cases of left and right rules for $\vee$,$\wedge$ to the reader as they are similar to the cases for implication rules. 
\\\\ \textbf{Right arrow rule}

Shape of $u$:
\[
\scalebox{0.7}{\begin{prooftree}
\hypo{\cns, \Gamma, x : A \vdash u_0 : B}
\infer1{\cns, \Gamma \vdash {\lfun_R(x,u_0)} : A \to  B}
\end{prooftree}}
\]
Suppose $\vec{w} \in \rset{\Gamma_f}$ but $\mcut(\vec{w},\vec{z},u_\infty) \notin \rset{(A \to B)_f}$. Since we have 
\[\mcut(\vec{w},\vec{z},u_\infty) \longrightarrow {\lfun_R}(x, \mcut(\vec{w},\vec{z},(u_0)_\infty))\]
it follows that there exists some $v \in \rset{A_f}$ such that $\mcut(v,x,\mcut(\vec{w},\vec{z},(u_0)_\infty)) \notin \rset{B_f}$. But we also have:
\[\mcut(v,x,\mcut(\vec{w},\vec{z},(u_0)_\infty)) \longrightarrow \mcut(v,\vec{w},x,\vec{z},(u_0)_\infty)\]
using that $v$ is an assumption-free proof, so it follows that $\mcut(v,\vec{w},x,\vec{z},(u_0)_\infty) \notin \rset{B_f}$. We set $u^* = u_0$, $f^* = f$. 
\\\\ \textbf{Left arrow rule}

Shape of $u$:
\[
\scalebox{0.7}{\begin{prooftree}
\hypo{\cns, \Gamma \vdash u_0 : A}
\hypo{\cns,  z : B, \Gamma \vdash u_1 : C}
\infer2{\cns, x : A \lfun B, \Gamma \vdash {\lfun_L}(x,z,u_0,u_1) : C}
\end{prooftree}}
\]

Suppose $v \in \rset{(A \to B)_f}$, $\vec{w} \in \rset{\Gamma_f}$ but $\mcut(v,\vec{w},x,\vec{z},u_\infty) \notin \rset{C_f}$. If $\mcut(\vec{w}, \vec{z},(u_0)_\infty) \notin \rset{A}$ then we are done as we can set $f^* = f$ and $u^* = u_0$. So suppose $\mcut(\vec{w}, \vec{z},(u_0)_\infty) \in \rset{A_f}$.  Since $v \in \rset{(A \to B)_f}$, there is some $y, v'$ such that $v \longrightarrow {\lfun_R}(y,v')$ and $\cut(s,y,v'_\infty) \in \rset{B_f}$ for all $s \in \rset{A_f}$. In particular, by our assumption, we have:
\[\cut(\mcut(\vec{w}, \vec{z},(u_0)_\infty),y,v'_\infty) \in \rset{B_f}\] 
Furthermore, we have:
\[
\begin{aligned}
\mcut(v,\vec{w},x,\vec{z},u_\infty) & = \mcut(v,\vec{w},x,\vec{z},{\lfun_L}(x,z,(u_0)_\infty,(u_1)_\infty)) \\
& \longrightarrow \mcut({\lfun_R}(y,v'),\vec{w},x,\vec{z},{\lfun_L}(x,z,(u_0)_\infty,(u_1)_\infty)) \\
& \longrightarrow \mcut(\cut(\mcut(\vec{w}, \vec{z},(u_0)_\infty),y,v'), \vec{w},z,\vec{z},(u_1)_\infty)
\end{aligned}
\]
Hence, by Proposition \ref{p:closure-of-candidates} we get
\[\mcut(\cut(\mcut(\vec{w}, \vec{z},(u_0)_\infty),y,v'), \vec{w},z,\vec{z},(u_1)_\infty) \notin \rset{C_f}\]
with $\cut(\mcut(\vec{w}, \vec{z},u_0),y,v') \in \rset{B_f}$, $\vec{w} \in \rset{\Gamma_f}$ as required. So we can set $f^* = f$ and $u^* = u_1$.   
\\\\ \textbf{Right $\mu$-rule}

Shape of $u$:
\[
\scalebox{0.7}{\begin{prooftree}
\hypo{\cns, \Gamma \vdash u_0 : B[\mu^\beta X. B/X]}
\infer1{\cns, \Gamma  \vdash \mu_R(\alpha,\beta,u_0) : \mu^\alpha X.B}
\end{prooftree}}
\]
Let $\vec{v}  \in \rset{\Gamma_f}$ be proofs  such that $\mcut (\vec{v},\vec{x},u_\infty) \notin \rset{(\mu^\alpha X. B)_f}$. We claim that $\mcut(\vec{v},\vec{x},(u_0)_\infty)\! \notin \rset{A[\mu^\beta X. A/X]_f}$. For suppose $\mcut(\vec{v},\vec{x},(u_0)_\infty) \in \rset{(A[\mu^\beta X. A/X])_f}$. If $\alpha$ is an ordinal variable, so is $\beta$ and $f(\beta) < f(\alpha)$ since $f$ is faithful. Hence:
\[(A[\mu^\beta X. A/X])_f = A_f[\ann{f(\beta)} X. A_f/X]\]
We get:
\[
\begin{aligned}
\mcut (\vec{v},\vec{x},u_\infty) & = \mcut (\vec{v}, \vec{x},\mu_R(\infty,\infty ,(u_0)_\infty))  \\
& \longrightarrow \mu_R(\infty,\infty,\mcut(\vec{v},\vec{x},(u_0)_\infty)) \\ 
&\in \rset{\ann{f(\alpha)} X. A_f} \\
& = \rset{(\mu^\alpha X. A)_f}
\end{aligned}
\]
which is a contradiction. On the other hand, if $\alpha = \infty$ then 
\[\rset{(\mu^\alpha X. A)_f} = \rset{\mu^\Omega X. A_f}\]
 We get:
\[
\begin{aligned}
& \mcut (\vec{v},\vec{x},u_\infty)  \\ & = \mcut (\vec{v}, \vec{x},\mu_R(\infty,\infty ,(u_0)_\infty))  \\
& \longrightarrow \mu_R(\infty,\infty,\mcut(\vec{v},\vec{x},(u_0)_\infty)) \\ 
&\in \rset{\mu^{f(\beta) + 1} X. A_f} & \\
& \subseteq \rset{\ann{\Omega + 1} X. A_f} &  \text{ $f(\beta) \leq \Omega$, Proposition \ref{p:large-ordinal} (Inflation)} \\
& = \rset{\mu^{\Omega} X. A_f} & \text{Proposition \ref{p:large-ordinal} (Closure)} \\
& = \rset{(\mu^\alpha X. A)_f} 
\end{aligned}
\]
which is again a contradiction. In each case we set $u^* = u_0$ and $f^* = f$. 
\\\\ \textbf{Left $\mu$-rule}

Shape of $u$:
\[
\scalebox{0.7}{\begin{prooftree}
\hypo{\cns +_\alpha \ordc,  y :  A[\mu^\ordc X.A/X], \Gamma \vdash u_0 : B}
\infer1{\cns,  x : \mu^\alpha X. A, \Gamma \vdash \mu_L(\alpha,\ordc, y, x,  u_0) : B}
\end{prooftree}}
\]
Let $v$, $\vec{w}$ be closed proofs with $v \in \rset{(\mu^\alpha X. A)_f}$, $ \vec{w} \in \rset{\Gamma_f}$. Then there is some ordinal $\xi$ (possibly $\Omega$) such that $(\mu^\alpha X. A)_f = \mu^\xi X.A_f$, and so there is some $\rho < \xi$ and some $v' \in \rset{A_f[\mu^\rho X. A_f/X]}$ such that $v \longrightarrow \mu_R(\infty, \infty, v')$. Let $u_0'$ denote the premiss of the root sequent in $u_\infty$, so that $u_0'$ has the single ordinal variable $c$ in its constraint. We get:
\[
\begin{aligned}
\mathsf{mcut}(v, \vec{w},x, \vec{z}, u_\infty) & =  \mathsf{mcut}(v, \vec{w}, x, \vec{z}, \mu_L(\infty, \ordc, y, x,  u_0')) \\
& \longrightarrow  \mathsf{mcut}(\mu_R(\infty,\infty, v'), \vec{w}, x, \vec{z}, \mu_L(\infty,\ordc, y, x, u_0')) \\
&  \longrightarrow \mathsf{mcut}(v', \vec{w}, y, \vec{z}, u_0'[c:=\infty]) \\
&  = \mathsf{mcut}(v', \vec{w}, y, \vec{z}, (u_0)_\infty)
\end{aligned}
\]
Set $f^*$ to be the extension of $f$ with $f^*(c) = \rho$, and set $u^* = u_0$. This works because, by the calculation above and by Proposition \ref{p:closure-of-candidates}, if $\mathsf{mcut}(v', \vec{w}, y, \vec{z}, (u_0)_\infty) \in \rset{B_{f^*}}$ then $\mathsf{mcut}(v, \vec{w},x, \vec{z}, u_\infty) \in \rset{B_{f^*}}$. But $\rset{B_{f^*}} = \rset{B_f}$ since the variable $c$ does not occur in $B$, and this contradicts our assumption. Since $c$ does not occur in  $\Gamma$ either we get $\vec{w} \in \rset{\Gamma_f} = \rset{\Gamma_{f^*}}$, and we have $v' \in  \rset{A_f[\mu^\rho X. A_f/X]} = \rset{(A[\mu^c X. A/X])_{f^*}}$ by assumption. 
\\\\ \textbf{Right $\nu$-rule}

Shape of $u$:
\[
\scalebox{0.7}{\begin{prooftree}
\hypo{\cns+_\alpha c, \Gamma \vdash u_0 : B[\nu^c X. B/X]}
\infer1{\cns, \Gamma  \vdash \mu_R(\alpha,\beta,u_0) : \nu^\alpha X.B}
\end{prooftree}}
\]
Let $\vec{v}  \in \rset{\Gamma_f}$ be proofs  such that $\mcut (\vec{v},\vec{x},u_\infty) \notin \rset{(\nu^\alpha X. B)_f}$.  We claim that there exists a faithful extension $f^*$ of $f$ such that  $\mcut(\vec{v},\vec{x},(u_0)_\infty) \notin \rset{A[\nu^c X. A/X]_{f^*}}$. Again, we make a case distinction as to whether $\alpha$ is an ordinal variable or $\infty$. In the former case, suppose $\mcut(\vec{v},\vec{x},(u_0)_\infty) \in \rset{(A[\nu^c X. A/X])_{f^*}}$ for every such extension; this means that  $(u_0)_\infty \in \rset{(A_f[\mu^\rho X. A_f/X])}$ for all $\rho < f(\alpha)$. But we have:
\[
\begin{aligned}
\mcut (\vec{v}, \vec{x},u_\infty) & =  \mcut (\vec{v}, \vec{x},\nu_R(\infty,c ,(u_0)_\infty))  \\
& \longrightarrow \nu_R(\infty,c,\mcut(\vec{v},\vec{x},(u_0)_\infty)) \\ 
\end{aligned}
\]
and we get  $\mcut (\vec{v}, \vec{x},\nu_R(\infty,c ,(u_0)_\infty))  \in \rset{(\nu^\alpha X. A)_f}$, which is a contradiction. 

On the other hand, if $\alpha = \infty$ then 
\[\rset{(\nu^\alpha X. A)_f} = \rset{\nu^\Omega X. A_f)} = \rset{\nu^{\Omega + 1} X.A_f}\]
by the Closure condition of Proposition \ref{p:large-ordinal}. 
We can now reason as in the previous case to find some $\rho < \Omega + 1$ such that $\mcut(\vec{v},\vec{x},(u_0)_\infty) \notin \rset{(A[\nu^c X. A/X])_{f^*}}$ if we set $f^*(c) = \rho$; since $\rho \leq \Omega$ this annotation is faithful. 
\\\\ \textbf{Left $\nu$-rule }

This is similar to the case of the left $\mu$-rule.
\\\\ \textbf{Rules for $\exists$ and $\forall$:}

These are very similar to the rules for $\mu$ and $\nu$. We present only the case for the right $\exists$-rule to illustrate this. In this case the shape of $u$ is:
\[
\scalebox{0.7}{\begin{prooftree}
\hypo{\cns, \Gamma \vdash u_0 : B[\alpha/c]}
\infer1{\cns, \Gamma  \vdash \exists_R(\alpha,\beta,u_0) : \exists c < \alpha :B}
\end{prooftree}}
\]
Let $\vec{v}  \in \rset{\Gamma_f}$ be proofs  such that $\mcut (\vec{v},\vec{x},u_\infty) \notin \rset{(\exists c < \alpha : B)_f}$. We claim that $\mcut(\vec{v},\vec{x},(u_0)_\infty) \notin \rset{A[\beta/c]_f}$. For suppose $\mcut(\vec{v},\vec{x},(u_0)_\infty) \in \rset{A[\beta/c]_f}$. If $\alpha$ is an ordinal variable, so is $\beta$ and $f(\beta) < f(\alpha)$ since $f$ is faithful. We have:
\[A[\beta/c]_f = A_f[f(\beta)/c]\]
and hence:
\[\mcut(\vec{v},\vec{x},(u_0)_\infty) \in \rset{A_f[f(\beta)/c]}\]
From this we get:
\[
\begin{aligned}
\mcut(\vec{v}, \vec{x}, u_\infty) & = \mcut (\vec{v}, \vec{x},\exists_R(\infty,\infty ,(u_0)_\infty))  \\
& \longrightarrow \exists_R(\infty,\infty,\mcut(\vec{v},\vec{x},(u_0)_\infty)) \\ 
&\in \rset{\exists c < f(\alpha). A_f} \\
& = \rset{(\exists c < \alpha. A)_f}
\end{aligned}
\]
which is a contradiction. On the other hand, if $\alpha = \infty$ then 
\[\rset{(\exists c < \alpha: A)_f} = \rset{\exists c < \infty : A_f}\]
As before, we have: \[\mcut(\vec{v},\vec{x},(u_0)_\infty) \in \rset{A_f[f(\beta)/c]}\]
 We get:
\[
\begin{aligned}
\mcut(\vec{v}, \vec{x}, u_\infty) & = \mcut (\vec{v}, \vec{x},\exists_R(\infty,\infty ,(u_0)_\infty))  \\
 & \longrightarrow \exists_R(\infty,\infty,\mcut(\vec{v},\vec{x},(u_0)_\infty)) \\ 
&\in \rset{\exists c < \infty : A_f} &    \text{ $f(\beta) \leq \Omega$ }\\
& = \rset{(\exists c < \alpha : A)_f} 
\end{aligned}
\]
which is again a contradiction. In each case we set $u^* = u_0$ and $f^* = f$. 
\\\\ \textbf{Cut rule}

Shape of $u$:
\[
\scalebox{0.7}{\begin{prooftree}
\hypo{\cns, \Gamma_1 \vdash u_1 : A_1}
\hypo{\hdots}
\hypo{\cns, \Gamma_n \vdash u_n : A_n}
\hypo{\cns, x_1 : A_1,\hdots , x_n : A_n, \Delta \vdash v : B}
\infer4{\cns, \Gamma_1,\hdots,\Gamma_n,\Delta \vdash \mcut(u_1,\hdots, u_n, x_1,\hdots, x_n, v) : B}
\end{prooftree}}
\]
Let $\vec{w}_i \in \rset{\Gamma_i}$ and $\vec{w}_{n +1} \in \rset{\Delta}$.  Suppose that \[\mcut(\vec{w}_1,\hdots, \vec{w}_n,\vec{w}_{n+1}, \vec{z}_1,\hdots,\vec{z}_n,\vec{z}_{n+1},\mcut(u_1,\hdots, u_n, x_1,\hdots, x_n, v)) \notin \rset{B}\]
A cut permutation reduces this term to:
\[\mcut(\mcut(\vec{w}_1,\vec{z}_1, u_1)\hdots, \mcut(\vec{w}_1,\vec{z}_n, u_n), \vec{w}_{n+1},  x_1,\hdots, x_n, \vec{z}_1,\hdots,\vec{z}_n,\vec{z}_{n+1}, v) \]
If $\mathsf{mcut}(\vec{w}_i,\vec{z}_i, u_i) \notin \rset{A_i}$ for some $i$ then the premiss sought after is $u_i$ with arguments $\vec{w}_i$. Otherwise we have $\mathsf{mcut}(\vec{w}_i,\vec{z}_i, u_i) \in \rset{A_i}$ for all $i$ and the sought after premiss is $v$  with arguments $\mathsf{mcut}(\vec{w}_i,\vec{z}_i, u_i)$ for $i \in \{1,\hdots, n\}$ together with $\vec{w}_{n+1}$. 
\\\\ \textbf{Weakening rule}

Shape of $u$:
\[
\scalebox{0.7}{\begin{prooftree}
\hypo{\cns, \Gamma \vdash u_0 : B}
\infer1{\cns, x : A, \Gamma \vdash W(x,u_0) : B}
\end{prooftree}}
\]
Suppose $v \in \rset{A}, \vec{w} \in \rset{\Gamma}$ but $\mcut(v, \vec{w}, x, \vec{z},u) \notin \rset{B}$. We have 
\[\mcut(v, \vec{w}, x, \vec{z},u) \longrightarrow \mcut(\vec{w},\vec{z},u_0)\]
so we immediately get $\mcut(\vec{w},\vec{z},u_0) \notin \rset{B}$.
\\\\
\textbf{Contraction rule }

Shape of $u$:
\[
\scalebox{0.7}{\begin{prooftree}
\hypo{\cns, y_0: A, y_1 : A, \Gamma \vdash u_0 : B}
\infer1{\cns, x : A, \Gamma \vdash C(x,y_0,y_1,u_0) : B}
\end{prooftree}}
\]
Suppose $v \in \rset{A}$, $\vec{w} \in \rset{\Gamma}$ but $\mcut(v, \vec{w}, x, \vec{z},u) \notin \rset{B}$. We have 
\[\mcut(v,\vec{w}, x,\vec{z},u) \longrightarrow \mcut(v,v,\vec{w},y_0,y_1,\vec{z},u_0)\]
so we immediately get $\mcut(v,v,\vec{w},y_0,y_1,\vec{z},u_0) \notin \rset{B}$.
\end{proof}

The next proposition makes use of Proposition \ref{p:one-step-descent} to relate the validity condition of proofs to computability:

\begin{prop}
\label{p:infinite-descent}
Let $u$ be a valid proof of a sequent of the form $\Gamma \vdash A$ with trivial constraint and let $\vec{v} \in \rset{\Gamma}$.  Then $\mcut(\vec{v},\vec{x},u) \in \rset{A}$. 
\end{prop}
\begin{proof}
Suppose $\mcut(\vec{v},\vec{x},u)  \notin \rset{A}$. By repeated use of Proposition \ref{p:one-step-descent} we can inductively construct an infinite branch $s_0s_1s_2\hdots$ in $u$, together with substitutions  $\sigma_0,\sigma_1,\sigma_2\hdots$  of ordinals for ordinal variables such that:
\begin{itemize}
\item for each $i < \omega$ and all ordinal variables $a,b$ such that $\cns(s_i) \Vdash a < b$ we have $\sigma_i(a) < \sigma_i(b)$.
\item $\sigma_j(a) = \sigma_i(a)$ for all $a$ appearing in $\cns(s_i)$ whenever $i < j$. 
\end{itemize}
Here, $\cns(s_i)$ denotes the constraint of the sequent labelling the vertex $s_i$ in the proof tree. The existence of an infinite descending chain of ordinal variables on $s_0s_1s_2\hdots$ then produces an infinite chain of ordinals $\xi_0 > \xi_1 > \xi_2,\hdots$, which contradicts wellfoundedness of the ordinals. 
\end{proof}

What follows is our first main result, showing that every closed valid proof is computable. 

\begin{thm}
\label{t:soundness}
For any closed, valid proof $u$ of $\vdash A$, we have $u \in \rset{A}$. 
\end{thm}

\begin{proof}
Direct corollary to Proposition \ref{p:infinite-descent}. 
\end{proof}

\subsection{Computation with finite data}

A first consequence of our main result, Theorem \ref{t:soundness}, is that proofs representing functions on natural numbers are normalizing. We can generalize this a bit; we define the \emph{finitary formulas} by the following grammar: 
\[\mathcal{F}_{\mathsf{ID}} \ni A:= \top  \mid A \lsum A \mid A \lprod A \mid \mu X.B\]
\[\mathcal{F}^+(X) \ni B := A \mid X \mid \top  \mid B \lsum B \mid B \lprod B\]

\begin{defi}
A proof $u$ is said to be in \emph{normal form} if there is no $v\neq u$ such that $u \longrightarrow v$. A proof $u$ is said to be \emph{normalizing} if $u \longrightarrow v$ for some $v$ in normal form. 
\end{defi}
Given a finitary formula $A \in \mathcal{F}_\mathsf{ID}$, we denote by $\mathsf{N}_A$ the set of normalizing proofs with end sequent $\vdash A$.

\begin{thm}
\label{t:computable-implies-normalizing}
If $A$ is a finitary formula, then $\rset{A} \subseteq \mathsf{N}_A$.
\end{thm}

\begin{proof}
By induction on the complexity of $A$.  The proof is straightforward and so we omit the details. 
\end{proof}

As a consequence of Theorem \ref{t:computable-implies-normalizing},  we get: 

\begin{thm}
Any closed proof of the sequent $N^\infty,\hdots,N^\infty \vdash N^\infty$ (with trivial constraint) represents a unique $k$-ary function on natural numbers.
\end{thm}

Here, we recall that by a proof we mean a \emph{valid} derivation.  For an arbitrary derivation with end sequent $ \vdash N^\infty$, normalization may not terminate. So derivations generally represent partial functions, and valid proofs represent total functions.

To be precise, we first note that we may identify a natural number with a closed proof of $\vdash N^\infty$ in normal form; explicitly we define a map $n \mapsto \underline{n}$ from $\mathbb{N}$ to proofs by:
\[\underline{0} = \begin{prooftree}
\hypo{}
\infer1{\vdash \top}
\infer1{\vdash \top \vee N^\infty}
\infer1{\vdash N^\infty}
\end{prooftree}
\qquad
\underline{n + 1} = \begin{prooftree}
\hypo{\underline{n}}
\infer1{\vdash N^\infty}
\infer1{\vdash \top \vee N^\infty}
\infer1{\vdash N^\infty}
\end{prooftree}
\]
This map is surjective onto the set of normal form proofs of $\vdash N^\infty$ since we require of a closed proof that the constraint is empty, so that the only ordinal term $\alpha$ for which the constraint proves $\alpha < \infty$ is $\infty$ itself, hence there is only one possible instance of the $\mu_R$-rule available with the sequent $\vdash N^\infty$ as conclusion. Now, given a proof $u$ of $N^\infty,\hdots,N^\infty \vdash N^\infty$ and natural numbers $n_1,\hdots n_k$ represented as closed proofs $\underline{n}_1,\hdots,\underline{n}_k$ of the sequent $\vdash N^\infty$, the proof $\mcut(\underline{n}_1,\hdots,\underline{n}_k,x_1,\hdots,x_k,u)$ normalizes to $\underline{m}$ for a unique natural number $m$. Of course, uniqueness follows from  Proposition \ref{p:confluence}.

\section{Categorical semantics}

In this final section we further examine the computability semantics of proofs by deriving a categorical model with closed formulas as objects and proofs as arrows. For an early categorical semantics of non-wellfounded proofs, see \cite{santocanale2002calculus,fortier2013cuts}. For a recent study of categorical models of non-wellfounded proofs in the context of linear logic, see \cite{ehrhard2025denotation}.

\subsection{Construction of the category of proofs}
We begin by introducing an equivalence relation on proofs.

\begin{defi}
Given a formula $A \in \mathcal{I}_\mathsf{ID}$ and closed proofs $u, v$ of the sequent $\vdash \underline{A}$, we define a relation $u \equiv v \md A$ by induction on the structure of $A$ as follows:

\begin{description}
\item[$\top$] $u \equiv v \md \top$ iff $u, v \longrightarrow \mathsf{ax}_\top$.
\item[$\lsum$] $u \equiv v \md A_0 \lsum A_1$ iff $u \longrightarrow \lsum_R^i(u')$ and $v \longrightarrow \lsum_R^i(v')$ for $u' \equiv v' \md A_i$.
\item[$\lprod$] $u \equiv v \md A \lprod B$ iff $u \longrightarrow \lprod_R(u'_0,u'_1)$ and $v \longrightarrow \lsum_R^i(v'_0,v'_1)$ for $u'_i \equiv v'_i \md A_i$.
\item[$\lfun$]  $u \equiv v \md A \lfun B$ iff $u \longrightarrow {\lfun}_R(x,u')$ and $v \longrightarrow {\lfun}_R(y,v')$ for $u',v'$ such that, for all closed proofs $w$ of $\underline{A}$, $\mcut(w,x,u') \equiv \mcut(w,y,v') \md B$.
\item[$\mu^\infty$] $u \equiv v \md \mu^\infty X. A$ iff $u \equiv v \md \mu^\xi X. A$ for some ordinal $\xi \leq \Omega$.
\item[$\nu^\infty$] $u \equiv v \md \nu^\infty X. A$ iff $u \equiv v \md \nu^\xi X. A$ for every ordinal $\xi \leq \Omega$.
\item[$\mu^\xi$] $u \equiv v \md \mu^\xi X. A$ iff $u \longrightarrow \mu_R(\infty, \infty,u')$ and $v \longrightarrow \mu_R(\infty, \infty, v')$ for $u',v'$ such that $u' \equiv v' \md A[\mu^\rho X. A/X]$ for some $\rho < \xi$.
\item[$\nu^\xi$] $u \equiv v \md \nu^\xi X. A$ iff $u \longrightarrow \nu_R(\infty, a,u')$ and $v \longrightarrow \nu_R(\infty, b, v')$ for $u',v'$ such that $u'[a:=\infty] \equiv v'[b:= \infty] \md A[\nu^\rho X. A/X]$ for all $\rho < \xi$.
\item[$\exists a < \infty$] $u \equiv v \md \exists a < \infty: A$ iff$u \longrightarrow \exists_R(\infty, \infty,u')$ and $v \longrightarrow \exists_R(\infty, \infty, v')$ for $u',v'$ such that $u' \equiv v' \md A[\rho/c]$ for some $\rho < \Omega$.
\item[$\forall a < \infty$] $u \equiv v \md \forall a < \infty: A$ iff $u \longrightarrow \forall_R(\infty, a,u')$ and $v \longrightarrow \forall_R(\infty, b, v')$ for $u',v'$ such that $u'[a:=\infty] \equiv v'[b:= \infty] \md A[\rho/c]$ for all $\rho < \Omega$.
\item[$\exists c < \xi$] $u \equiv v \md \exists c < \xi: A$ iff $u \longrightarrow \exists_R(\infty, \infty,u')$ and $v \longrightarrow \exists_R(\infty, \infty, v')$ for $u',v'$ such that $u' \equiv v' \md A[\rho/c]$ for some $\rho < \xi$.
\item[$\forall c < \xi$] $u \equiv v \md \forall c < \xi: A$ iff $u \longrightarrow \forall_R(\infty, a,u')$ and $v \longrightarrow \forall_R(\infty, b, v')$ for $u',v'$ such that $u'[a:=\infty] \equiv v'[b:= \infty] \md A[\rho/c]$ for all $\rho < \xi$.
\end{description}
If $u,v$ are proofs of the sequent $\vdash A$ for $A \in \mathcal{L}_\mathsf{ID}$, we write $u \equiv v$ if $u \equiv v \md A$.
\end{defi}

\begin{prop}
The relation $\equiv$ is an equivalence relation over closed proofs.
\end{prop}
\begin{proof}
Symmetry is obvious, and transitivity follows from confluence (Proposition \ref{p:confluence}).  Finally, reflexivity follows from Theorem \ref{t:soundness}.
\end{proof}

Arrows will correspond to proofs of sequents of the form $A \vdash B$ with trivial constraint. 
For proofs $u,v$ of the sequent $A \vdash B$ we write $u \equiv v$ if, for all closed proofs $w$ of the sequent $\vdash A$, $\cut(w,u) \equiv \cut(w, v)$. Composition of arrows will by defined using cuts, and we need to show that the equivalence relation $\equiv$ is a congruence with respect to cuts.

\begin{prop}
\label{p:congruence}
Let $w$ be a proof of $A \vdash B$ and $u,v$ closed proofs of $A$.  If $u \equiv v$ then $\cut(u, w) \equiv \cut(v,w)$.
\end{prop}
To prove Proposition \ref{p:congruence}, we essentially mimic the proof of Theorem \ref{t:soundness}. We first prove:
\begin{prop}
\label{p:one-step-descent-cong}
For each proof $u$ of a sequent $\cns, \Gamma \vdash A$, each faithful annotation $f$ of the sequent and each pair of tuples $\vec{v},\vec{w} \in \rset{\Gamma_f}$, if
\[\vec{v} \equiv \vec{w} \md \Gamma_f \text{ but not } \mcut(\vec{v},\vec{x},u_\infty) \equiv \mcut(\vec{w},\vec{x},u_\infty) \]
then there exists a premiss $u^*$ of $u$ labelled $\cns^*,\Gamma^* \vdash ^*$, an extension $f^*$ of $f$ to $\cns^*$ and tuples of closed proofs $\vec{v}_*,\vec{w}_*$ such that:
\begin{itemize}
\item $\vec{v}_*  \equiv \vec{w}_* \md \Gamma_{f^*}$,
\item not $\mcut(\vec{v}_*,\vec{x}_*,u^*_{\infty}) \equiv \mcut(\vec{w}_*,\vec{x}_*,u^*_{\infty}) \md A^*_{f^*}$.
\end{itemize}
\end{prop}
Proposition \ref{p:congruence} can then be proved from this proposition similarly to how we proved Theorem \ref{t:soundness}. We omit the details since it the full proof would essentially be a step-by-step repetition of the proof of Theorem \ref{t:soundness} in a different guise.

We can now define a category $\mathbb{C}$ whose objects are closed formulas of $\mathcal{L}_{\mathsf{ID}}$, and whose arrows are equivalence classes of proofs $A \vdash B$. Composition $v \circ u$ of arrows $u : A \to B$ and $v  : B \to C$ is given by the equivalence class of $\cut(u, x, v)$ (where $x$ is an arbitrary variable assigned to the left-hand formula of the sequent $B \to C$ in a term representation of $v$). This is well defined by Proposition \ref{p:congruence} and the definition of equality of arrows. Associativity of composition is given by Proposition \ref{p:associativity} below. 

\begin{prop}
\label{p:asslemma}
Let $u_0,u_1$ be proofs of the sequent $A \vdash B$ where $A,B$ are closed formulas in $\mathcal{L}_{\mathsf{ID}}$, and suppose for all closed proofs $v$ of $\vdash A$ there exists a proof $w$ such that $\mcut(v,x,u_0) \longrightarrow w$ and $\mcut(v,x,u_1) \longrightarrow w$. Then $u_0 \equiv u_1$. 
\end{prop}
\begin{proof}
By the definition of $\equiv$, it suffices to prove that $\mcut(v,x,u_0) \equiv \mcut(v,x,u_1)$ provided that there is some proof $w$ such that $\mcut(v,x,u_0) \longrightarrow w$ and $\mcut(v,x,u_1) \longrightarrow w$.

We prove something stronger: if $A$ is any closed formula of $\mathcal{I}_{\mathsf{ID}}$, $u \in \rset{A}$, and $v$ is any closed proof of  $\vdash \underline{A}$ for which there exists a proof $w$ with $u \longrightarrow w$ and $v \longrightarrow w$, then $u \equiv v \md A$. The proof proceeds by a straightforward induction on the wellfounded order $\prec$. We provide only the case where $A$ is of the form $\exists c < \xi : B$ to illustrate how the induction works, leaving the other cases to the reader. In this case, $\underline{\exists c < \xi : B} = \exists c < \infty : \underline{B}$. By assumption $u \in \rset{\exists c < \xi : B}$, so there exists some ordinal $\rho_0 < \xi$ and some $u' \in \rset{B[\rho_0/c]}$ such that $u \longrightarrow \exists R(\infty,\infty,u')$. By Theorem \ref{t:soundness}, there is some ordinal $\rho_1 \leq \Omega$ and some $v' \in \rset{B[\rho_1/c]}$ such that $v \longrightarrow \exists_R(\infty, \infty, v')$. Since $u, v \longrightarrow w$, repeated uses of Proposition \ref{p:confluence} gives some proof $s$ with $ \exists_R(\infty, \infty, u') \longrightarrow s$ and $ \exists_R(\infty, \infty, v') \longrightarrow s$, and it is clear that $s$ must be of the form $ \exists_R(\infty, \infty, s')$ where $u' \longrightarrow s'$ and $v' \longrightarrow s'$. Since $u' \in \rset{B[\rho_0/c]}$ and $v'$ is a closed proof of $\vdash \rset{\underline{B}[\infty/c]}$, i.e. of $\vdash \underline{B[\rho_0/c]}$, the induction hypothesis applied to $B[\rho_0/c] \prec \exists c < \xi : B$ gives $u' \equiv v' \md B[\rho_0/c]$. Since $\rho_0 < \xi$ this yields \[u \equiv v \md \exists c < \xi : B\] as required. 
\end{proof}

\begin{prop}
\label{p:associativity}
Composition of arrows is associative. 
\end{prop}
\begin{proof}
We want to show that for given proofs $v_0,v_1,v_2$ of sequents $C \vdash D$, $B \vdash C$ and $A \vdash B$ respectively, we have  $(v_0 \circ v_1) \circ v_2 \equiv  v_0 \circ(v_1 \circ v_2)$. To prove this, we need to show that for any closed proof $u$ of the sequent $\vdash A$, we have:
\[\cut(u,(v_0 \circ v_1) \circ v_2) \equiv \cut(u,v_0 \circ(v_1 \circ v_2)) \md D\] 
To prove this, by Proposition \ref{p:asslemma} it suffices to show that these two terms reduce to the same proof.  The term  $\cut(u,(v_0 \circ v_1) \circ v_2)$ is the following proof:
\[
\scalebox{0.7}{
\begin{prooftree}
\hypo{u}
\infer1{\vdash A}
\hypo{v_2}
\infer1{A \vdash B}
\hypo{v_1}
\infer1{B \vdash C}
\hypo{v_0}
\infer1{C \vdash D}
\infer2[$cut$]{B \vdash D}
\infer2[$cut$]{A \vdash D}
\infer2[$cut$]{\vdash D}
\end{prooftree}
}
\]
A cut permutation gives:
\[
\scalebox{0.7}{
\begin{prooftree}
\hypo{u}
\infer1{\vdash A}
\hypo{v_2}
\infer1{A \vdash B}
\infer2[$cut$]{\vdash B}
\hypo{v_1}
\infer1{B \vdash C}
\hypo{v_0}
\infer1{C \vdash D}
\infer2[$cut$]{B \vdash D}
\infer2[$cut$]{\vdash D}
\end{prooftree}
}
\]
A second cut permutation gives:
\[
\scalebox{0.7}{
\begin{prooftree}
\hypo{u}
\infer1{\vdash A}
\hypo{v_2}
\infer1{A \vdash B}
\infer2[$cut$]{\vdash B}
\hypo{v_1}
\infer1{B \vdash C}
\infer2[$cut$]{\vdash C}
\hypo{v_0}
\infer1{C \vdash D}
\infer2[$cut$]{\vdash D}
\end{prooftree}
}
\]
On the other hand, the term $\cut(u,v_0 \circ(v_1 \circ v_2)) $ is:
\[
\scalebox{0.7}{
\begin{prooftree}
\hypo{u}
\infer1{\vdash A}
\hypo{v_2}
\infer1{A \vdash B}
\hypo{v_1}
\infer1{B \vdash C}
\infer2[$cut$]{A \vdash C}
\hypo{v_0}
\infer1{C \vdash D}
\infer2[$cut$]{A \vdash D}
\infer2[$cut$]{\vdash D}
\end{prooftree}
}
\]
Cut permutation gives:
\[
\scalebox{0.7}{
\begin{prooftree}
\hypo{u}
\infer1{\vdash A}
\hypo{v_2}
\infer1{A \vdash B}
\hypo{v_1}
\infer1{B \vdash C}
\infer2[$cut$]{A \vdash C}
\infer2[$cut$]{\vdash C}
\hypo{v_0}
\infer1{C \vdash D}
\infer2[$cut$]{\vdash D}
\end{prooftree}
}
\]
A second cut permutation gives, once again:
\[
\scalebox{0.7}{
\begin{prooftree}
\hypo{u}
\infer1{\vdash A}
\hypo{v_2}
\infer1{A \vdash B}
\infer2[$cut$]{\vdash B}
\hypo{v_1}
\infer1{B \vdash C}
\infer2[$cut$]{\vdash C}
\hypo{v_0}
\infer1{C \vdash D}
\infer2[$cut$]{\vdash D}
\end{prooftree}
}
\]
and we are done. 
\end{proof}
The identity arrow for object $A$ is of course just the proof consisting of the identity axiom for $A$. We shall not be pedantic about notationally distinguishing between proofs and their equivalence classes; it is to be understood that everything in this section is ``up to equivalence''.  

Our next aim is to show that least and greatest fixpoint formulas appears as initial algebras and final coalgebras in the category $\mathbb{C}$. For this purpose we need to show that formulas in $\mathcal{L}^+(X)$ correspond to endofunctors on $\mathbb{C}$. The construction essentially follows \cite{curzi2023computational}.

Given a proof $u$ of $B \vdash C$, and a formula $A \in \mathcal{L}^+(X)$, we define the proof $Au$ by induction on the complexity of $A$. If $X$ does not occur free in $A$ we set $Au$ to simply be $\mathsf{id}_A$. Other cases are as follows:
\paragraph{Case $A = X$:}
Set  $Au = u$. 
\paragraph{Case $A = A_0 \lprod A_1$:}
Set $Au$ to be:
\[
\scalebox{0.7}{\begin{prooftree}
\hypo{u}
\infer1{B \vdash C}
\infer1[$A_0$]{A_0[B/X] \vdash A_0[C/X]}
\hypo{u}
\infer1{B \vdash C}
\infer1[$A_1$]{A_1[B/X] \vdash A_1[C/X]}
\infer2{A_0[B/X],A_1[B/X] \vdash A_0[C/X] \lprod A_1[C/X]}
\infer1{A_0[B/X] \lprod A_1[B/X] \vdash A_0[C/X] \lprod A_1[C/X]}
\end{prooftree}}
\]

\paragraph{Case $A = A_0 \lsum A_1$:}
Set $Au$ to be:
\[
\scalebox{0.7}{\begin{prooftree}
\hypo{u}
\infer1{B \vdash C}
\infer1[$A_0$]{A_0[B/X] \vdash A_0[C/X]}
\infer1{A_0[B/X] \vdash A_0[C/X] \lsum A_1[C/X]}
\hypo{u}
\infer1{B \vdash C}
\infer1[$A_1$]{A_1[B/X] \vdash A_1[C/X]}
\infer1{A_1[B/X] \vdash A_0[C/X] \lsum A_1[C/X]}
\infer2{A_0[B/X] \lsum A_1[B/X] \vdash A_0[C/X] \lsum A_1[C/X]}
\end{prooftree}}
\]

\paragraph{Case $A = A_0 \lfun A_1$:}
Set $Au$ to be:
\[
\scalebox{0.7}{\begin{prooftree}
\hypo{u}
\infer1{B \vdash C}
\infer1[$A_0$]{A_0[B/X] \vdash A_0[C/X]}
\hypo{u}
\infer1{B \vdash C}
\infer1[$A_1$]{A_1[B/X] \vdash A_1[C/X]}
\infer2{A_0[B/X] \lfun A_1[B/X], A_0[C/X] \vdash A_1[C/X]}
\infer1{A_0[B/X] \lfun A_1[B/X] \vdash A_0[C/X] \lfun A_1[C/X]}
\end{prooftree}}
\]

\begin{prop}
The map sending (the equivalence class of) $u$ to (the equivalence class of) $Au$ is an endofunctor on $\mathbb{C}$.
\end{prop}
\begin{proof}
Routine, using the reduction rules.
\end{proof}

We also need the following somewhat technical result:

\begin{prop}
\label{p:A-lemma}
\label{p:A-lemma-nu}
Let $A$ be a formula in $\mathcal{L}^+(X)$ and let $B,C$ be formulas in $\mathcal{I}_\mathsf{ID}$. Let $u \in \rset{A[B/X]}$ and let $v,w$ be proofs such that $\mcut(u_0,v) \equiv \mcut(u_0,w) \md C$ for all $u_0 \in \rset{B}$. Then:  \[\mcut(u,x,Av) \equiv \mcut(u,x,Aw) \md A[C/X]\]
\end{prop}
\begin{proof}
By induction on the formula $A$. The straightforward but tedious argument is left to the reader. 
\end{proof}

\paragraph{A remark on the semantics}
The categorical semantics we have provided here is constructed from the syntax, where objects of the category are simply formulas and arrows are certain congruence classes of proofs. This may give the impression that this is just a term model, an initial semantics in which are just terms up to convertibility \cite{jacobs1999categorical}. But this is not the case; first, we have not provided strong enough reduction rules for a full cut elimination result here, so the relation of convertibility is quite weak. But even if we had done so, the congruence relation we have used for the construction of the semantics remains much stronger than convertibility. The idea is that two proofs of $A \vdash B$ should be identified as the same arrow whenever they represent (extensionally) the same function from type $A$ to type $B$. 

For example, consider the following two proofs:
\[
\scalebox{0.7}{
\begin{prooftree}
\hypo{}
\infer1{a < \infty \vdash \top}
\infer1{a < \infty \vdash \top \vee N^\infty}
\infer1{a < \infty  \vdash N^\infty}
\infer1{a < \infty \vdash \top \vee N^\infty}
\infer1{a < \infty  \vdash N^\infty}
\infer1{a < \infty, \top  \vdash N^\infty  }
\hypo{b < a < \infty, \top \vee N^b \vdash N^\infty\; \dagger}
\infer1{a < \infty, N^a \vdash N^\infty}
\infer1{a < \infty, N^a \vdash \top \vee  N^\infty}
\infer1{a < \infty, N^a \vdash N^\infty}
\infer2{a < \infty, \top \vee N^a \vdash N^\infty\; \dagger}
\infer1{N^\infty \vdash N^\infty}
\end{prooftree}
\qquad\qquad
\begin{prooftree}
\hypo{}
\infer1{a < \infty \vdash \top}
\infer1{a < \infty \vdash \top \vee N^\infty}
\infer1{a < \infty  \vdash N^\infty}
\infer1{a < \infty, \top  \vdash N^\infty  }
\hypo{b < a < \infty, \top \vee  N^b \vdash N^\infty\; \dagger}
\infer1{a < \infty,  N^a \vdash N^\infty}
\infer2{a < \infty, \top \vee N^a \vdash N^\infty \; \dagger}
\infer1{N^\infty \vdash N^\infty}
\infer1{N^\infty \vdash \top \vee N^\infty}
\infer1{N^\infty \vdash N^\infty}
\end{prooftree}
}
\]
Both these proofs are cut free, and so are in normal form with respect to the reduction rules we have provided, and with respect to any reasonable set of reduction rules (including permutations of left rules). So in an initial semantics, where terms are identified just when they are convertible to each other, these two proofs would get distinct interpretations. But in our categorical model, they get the same interpretation.  This is because the congruence relation is defined in terms of what happens when we supply arguments to the function that a proof represents, via a cut. In this case, cutting either of the two proofs above with a proof representing a natural number $n$ and normalizing would lead to the same normal form, namely the proof representing the successor $n+1$. For example, cutting either of the two with the proof:
\[
\scalebox{0.7}{
\begin{prooftree}
\hypo{}
\infer1{\vdash \top}
\infer1{\vdash \top \vee N^\infty}
\infer1{\vdash N^\infty}
\infer1{\vdash \top \vee N^\infty}
\infer1{\vdash N^\infty}
\end{prooftree}
}
\]
which represents the number $1$, normalization would yield the normal form:
\[
\scalebox{0.7}{
\begin{prooftree}
\hypo{}
\infer1{\vdash \top}
\infer1{\vdash \top \vee N^\infty}
\infer1{\vdash N^\infty}
\infer1{\vdash \top \vee N^\infty}
\infer1{\vdash N^\infty}
\infer1{\vdash \top \vee N^\infty}
\infer1{\vdash N^\infty}
\end{prooftree}
}
\]
representing the number $2$.
We may say that the interpretation of either proof just \emph{is} the successor function on natural numbers. But the differences between the two proofs are not superficial - they each compute the successor function in two different (both rather roundabout) ways. The left-hand proof computes the successor recursively: given $n$ as argument, return $1$ if $n = 0$, and otherwise recursively call the function on $n-1$ and take the successor of that. The right-hand proof instead does this: given the argument $n$, the subproof beginning at the third line from the bottom returns $n$ by computing the identity function recursively, and then the value $n$ is incremented once in the last two lines of the whole proof.

\subsection{Initial algebras}

Fix a formula $A \in \mathcal{L}^+(X)$. We want to show that $\mu^\infty X. A$ is the initial algebra for the corresponding endofunctor on $\mathbb{C}$. First of all we need to identify an arrow $i : A[\mu^\infty X. A/X] \to \mu^\infty X. A$. We define $i$ to be:
\[
\scalebox{0.7}{\begin{prooftree}
\hypo{}
\infer1{A[\mu^\infty X. A/X] \vdash A[\mu^\infty X. A/X]}
\infer1{A[\mu^\infty X. A/X] \vdash \mu^\infty X. A}
\end{prooftree}}
\]

We first check that $i$ is an isomorphism as expected. The inverse $i^{-1}$ is defined to be:
\[
\scalebox{0.7}{\begin{prooftree}
\hypo{}
\infer1{b < a < \infty, A[\mu^b X.A/X] \vdash A[\mu^b X. A/X]}
\infer1{b < a < \infty, A[\mu^b X.A/X] \vdash \mu^\infty X. A}
\infer1{a < \infty, \mu^a X.A \vdash \mu^\infty X. A}
\infer1[$A$]{a < \infty, A[\mu^a X. A/X] \vdash A[\mu^\infty X. A/X]}
\infer1{\mu^\infty X. A \vdash A[\mu^\infty X. A/X]}
\end{prooftree}}
\]

It is not hard to check that, for each closed proof $w_0$ of $\mu^\infty X. A$ and each closed proof $w_1$ of $A[\mu^\infty X. A/X]$, we have:
\[\mcut(w_0, x, \mcut(i^{-1}, y, i)) \equiv w_0\]
and 
\[\mcut(w_1, x, \mcut(i, y, i^{-1})) \equiv w_1\]
so that $i \circ i^{-1} = \mathsf{id}_{\mu^\infty X.A}$ and $i^{-1} \circ i = \mathsf{id}_{A[\mu^\infty X. A/X]}$.

We now check that any arrow $u : A[B/X] \to B$ lifts to a unique $A$-algebra morphism $\widehat{u} : \mu^\infty X. A \to B$. If $u : A[B/X] \to B$ then we can define the map $\widehat{u} : \mu^\infty X. A \to B$ as follows:
\[
\scalebox{0.7}{\begin{prooftree}
\hypo{b < a, A[\mu^b X. A/X]\vdash B \; \dagger}
\infer1{\mu^a X.A \vdash B}
\infer1[$A$]{A[\mu^a X.A/X] \vdash A[B/X]}
\hypo{u}
\infer1{A[B/X] \vdash B}
\infer2{A[\mu^a  X. A/X] \vdash B \; \dagger}
\infer1{\mu^\infty  X. A \vdash B}
\end{prooftree}}
\]

We now check that this is indeed an algebra morphism, i.e. that $u \circ A \widehat{u} = \widehat{u} \circ i$, as in the following commutative diagram:
\[
\xymatrix{
A[\mu^\infty X. A/X] \ar[rr]^{A\widehat{u}} \ar[d]_{i} & & A[B/X] \ar[d]^{u} \\
\mu^\infty X. A \ar[rr] _{\widehat{u}} & & B  \\
}
\]

The map $ u \circ A\widehat{u} $ is represented by the following proof: 
\[
\scalebox{0.7}{\begin{prooftree}
\hypo{b < a, A[\mu^b X. A/X]\vdash B \; \dagger}
\infer1{\mu^a X.A \vdash B}
\infer1[$A$]{A[\mu^a X.A/X] \vdash A[B/X]}
\hypo{u}
\infer1{A[B/X] \vdash B}
\infer2{A[\mu^a  X. A/X] \vdash B \; \dagger}
\infer1{\mu^\infty  X. A \vdash B}
\infer1[$A$]{A[\mu^\infty X. A/X] \vdash A[B/X]}
\hypo{u}
\infer1{A[B/X] \vdash B}
\infer2{A[\mu^\infty X. A/X] \vdash B}
\end{prooftree}}
\]

Consider a closed proof $w$ of $\vdash A[\mu^\infty X. A/X] $, and consider the result of cutting  $\widehat{u} \circ i$ with $w$. We get:
\[
\scalebox{0.7}{\begin{prooftree}
\hypo{w}
\infer1{\vdash A[\mu^\infty X. A/X]}
\hypo{}
\infer1{A[\mu^\infty X. A/X] \vdash A[\mu^\infty X. A/X]}
\infer1{A[\mu^\infty X. A/X] \vdash \mu^\infty X. A}
\hypo{c < b < a, A[\mu^c X. A/X]\vdash B \; \dagger}
\infer1{b < a, \mu^b X.A \vdash B}
\infer1[$A$]{b < a, A[\mu^b X.A/X] \vdash A[B/X]}
\hypo{u}
\infer1{A[B/X] \vdash B}
\infer2{b < a, A[\mu^b  X. A/X] \vdash B \;\dagger}
\infer1{\mu^a X.A \vdash B}
\infer1[$A$]{A[\mu^a X.A/X] \vdash A[B/X]}
\hypo{u}
\infer1{A[B/X] \vdash B}
\infer2{A[\mu^a  X. A/X] \vdash B }
\infer1{\mu^\infty  X. A \vdash B}
\infer2{A[\mu^\infty X. A/X] \vdash B}
\infer2{\vdash B}
\end{prooftree}}
\]
 Applying a few cut reductions gives:
\[
\scalebox{0.7}{\begin{prooftree}
\hypo{w}
\infer1{\vdash A[\mu^\infty X. A/X]}
\hypo{c < b, A[\mu^c X. A/X]\vdash B \; \dagger}
\infer1{ \mu^b X.A \vdash B}
\infer1[$A$]{ A[\mu^b X.A/X] \vdash A[B/X]}
\hypo{u}
\infer1{A[B/X] \vdash B}
\infer2{ A[\mu^b  X. A/X] \vdash B \;\dagger}
\infer1{\mu^\infty X.A \vdash B}
\infer1[$A$]{A[\mu^\infty X.A/X] \vdash A[B/X]}
\hypo{u}
\infer1{A[B/X] \vdash B}
\infer2{A[\mu^\infty  X. A/X] \vdash B }
\infer2{\vdash B}
\end{prooftree}}
\]

But this is just the result of cutting $ u \circ A\widehat{u} $ with $w$, up to a renaming of ordinal variables, so clearly these proofs represent the same arrow.

We now check that $\widehat{u}$ is the \emph{unique} arrow with this property. Here, we use Theorem \ref{t:soundness}; if $u^*$ is a  proof of $\mu^\infty X. A \vdash B$ with the same property, and $w$ is a closed proof of $\vdash \mu^\infty X. A$, we want to show that $\mcut(w,x,u^*) \equiv \mcut(w,x,\widehat{u})$. Since $w$ is a closed proof, by Theorem \ref{t:soundness} we have $w \in \rset{\mu^\xi X. A}$  for some ordinal $\xi$. So we can proceed by induction on $\xi$; if $w \in \rset{\mu^\xi X. A}$ then $w \longrightarrow \mu_R(\infty,\infty,w')$ for some $w' \in \rset{A[\mu^\rho X. A/X]}$, where the induction hypothesis holds for $\rho < \xi$. We have:
\[
\begin{aligned}
u^* & = u^* \circ \mathsf{id}_{\mu^\infty X. A} \\
& = u^* \circ i \circ i^{-1} \\
& = u \circ A u^* \circ i^{-1}
\end{aligned}
\]
So $\mcut(w,x,u^*) $ is equivalent to:
\[
\scalebox{0.7}{\begin{prooftree}
\hypo{w}
\infer1{\vdash \mu^\infty X. A}
\hypo{}
\infer1{b < a < \infty, A[\mu^b X.A/X] \vdash A[\mu^b X. A/X]}
\infer1{b < a < \infty, A[\mu^b X.A/X] \vdash \mu^\infty X. A}
\infer1{a < \infty, \mu^a X.A \vdash \mu^\infty X. A}
\infer1[$A$]{a < \infty, A[\mu^a X. A/X] \vdash A[\mu^\infty X. A/X]}
\infer1{\mu^\infty X. A \vdash A[\mu^\infty X. A/X]}
\hypo{u^*}
\infer1{\mu^\infty X. A \vdash B}
\infer1[$A$]{A[\mu^\infty X. A/X] \vdash A[B/X]}
\hypo{u}
\infer1{A[B/X] \vdash B}
\infer2{A[\mu^\infty X. A/X] \vdash B}
\infer2{\mu^\infty X. A \vdash B}
\infer2{\vdash B}
\end{prooftree}}
\]
which reduces to:
\[
\scalebox{0.7}{\begin{prooftree}
\hypo{w'}
\infer1{\vdash A[\mu^\infty X. A/X]}
\hypo{}
\infer1{b <\infty, A[\mu^b X.A/X] \vdash A[\mu^b X. A/X]}
\infer1{b <\infty, A[\mu^b X.A/X] \vdash \mu^\infty X. A}
\infer1{ \mu^\infty X.A \vdash \mu^\infty X. A}
\infer1[$A$]{A[\mu^\infty X. A/X] \vdash A[\mu^\infty X. A/X]}
\infer2{ A[\mu^\infty X. A/X]}
\hypo{u^*}
\infer1{\mu^\infty X. A \vdash B}
\infer1[$A$]{A[\mu^\infty X. A/X] \vdash A[B/X]}
\infer2{\vdash A[B/X]}
\hypo{u}
\infer1{A[B/X] \vdash B}
\infer2{\vdash B}
\end{prooftree}}
\]

But here, the subproof rooted at the sequent $A[\mu^\infty X. A/X] \vdash A[\mu^\infty X. A/X]$ is easily seen to be equivalent to the identity proof for $A[\mu^\infty X. A/X]$. So this is equivalent to:
\[
\scalebox{0.7}{\begin{prooftree}
\hypo{w'}
\infer1{\vdash A[\mu^\infty X. A/X]}
\hypo{u^*}
\infer1{\mu^\infty X. A \vdash B}
\infer1[$A$]{A[\mu^\infty X. A/X] \vdash A[B/X]}
\infer2{\vdash A[B/X]}
\hypo{u}
\infer1{A[B/X] \vdash B}
\infer2{\vdash B}
\end{prooftree}}
\]

Using the induction hypothesis on $\rho$, it follows by Proposition \ref{p:A-lemma} that this is equivalent to:
\[
\scalebox{0.7}{\begin{prooftree}
\hypo{w'}
\infer1{\vdash A[\mu^\infty X. A/X]}
\hypo{\widehat{u}}
\infer1{\mu^\infty X. A \vdash B}
\infer1[$A$]{A[\mu^\infty X. A/X] \vdash A[B/X]}
\infer2{\vdash A[B/X]}
\hypo{u}
\infer1{A[B/X] \vdash B}
\infer2{\vdash B}
\end{prooftree}}
\]
But by similar reasoning, $\mcut(w,x,\widehat{u})$ is also equivalent to this proof, hence $u^* \equiv \widehat{u}$ as required.

\subsection{Final coalgebras}

For a greatest fixpoint $\nu^\infty  X. A$, we define the final coalgebra  map $i$ dually as follows:
\[
\scalebox{0.7}{\begin{prooftree}
\hypo{}
\infer1{ A[\nu^\infty X. A/X] \vdash A[\nu^\infty X. A/X]}
\infer1{\nu^\infty X. A \vdash A[\nu^\infty X.A/X] }
\end{prooftree}}
\]

Again, we check that this is an isomorphism. 
The inverse $i^{-1}$ of this map is:
\[
\scalebox{0.7}{\begin{prooftree}
\hypo{}
\infer1{b < a < \infty, A[\nu^b X. A/X] \vdash A[\nu^b X. A/X]}
\infer1{b < a < \infty, \nu^\infty X. A \vdash A[\nu^b X. A/X]}
\infer1{a < \infty, \nu^\infty X. A \vdash \nu^a X. A}
\infer1[$A$]{a < \infty, A[\nu^\infty X. A/X] \vdash A[\nu^a X. A/X]}
\infer1{A[\nu^\infty X. A/X] \vdash \nu^\infty X. A}
\end{prooftree}}
\]

Given a proof $u$ of $B \vdash A[B/X]$, we want to find an $A$-coalgebra morphism $\widetilde{u} : B \to \nu^\infty X. A$. We define $\widetilde{u}$ to be:
\[
\scalebox{0.7}{\begin{prooftree}
\hypo{u}
\infer1{B \vdash A[B/X]}
\hypo{b < a, B \vdash A\nu^b X. A/X] \;\dagger}
\infer1{B \vdash \nu^a X. A}
\infer1[$A$]{A[B/X] \vdash A[\nu^a X. A/X]}
\infer2{B \vdash A[\nu^a X. A/X] \;\dagger}
\infer1{B \vdash \nu^\infty X. A}
\end{prooftree}}
\]

We can show that $i \circ \widetilde{u} = A \widetilde{u} \circ u$ using a similar argument as in the previous section. As a commutative diagram this is:
\[
\xymatrix{
A[B/X] \ar[rr]^{A \widetilde{u}} & & A[\nu^\infty X. A/X] \\
B \ar[u]^{u} \ar[rr]_{\widetilde{u}} & & \nu^\infty X. A \ar[u]_{i} \\
}
\]
 We show that $\widetilde{u}$ is unique up to equivalence with this property. Suppose $i \circ u^* = Au^* \circ u$. Then:
\[
\begin{aligned}
u^* & =  i^{-1} \circ i \circ u^* \\
& = i^{-1} \circ A u^* \circ u
\end{aligned}
\]
Now, suppose $w$ is a closed proof of $B$. We want to show that $\mcut(w,x,u^*) \equiv \mcut(w,x, \widetilde{u})$, i.e. that:
\[\mcut(w,x,u^*) \equiv \mcut(w,x, \widetilde{u}) \md \nu^\xi X. A\]
for any given ordinal $\xi$. We assume inductively that the statement holds for all $\rho < \xi$.  When we cut $i^{-1} \circ Au ^* \circ u$ with $w$, we get:
\[
\scalebox{0.7}{\begin{prooftree}
\hypo{w}
\infer1{\vdash B}
\hypo{u}
\infer1{B \vdash A[B/X]}
\hypo{u^*}
\infer1{B \vdash \nu^\infty X. A}
\infer1[$A$]{A[B/X] \vdash A[\nu^\infty X.A/X]}
\hypo{}
\infer1{b < a < \infty, A[\nu^b X. A/X] \vdash A[\nu^b X. A/X]}
\infer1{b < a < \infty, \nu^\infty X. A \vdash A[\nu^b X. A/X]}
\infer1{a < \infty, \nu^\infty X. A \vdash \nu^a X. A}
\infer1[$A$]{a < \infty, A[\nu^\infty X. A/X] \vdash A[\nu^a X. A/X]}
\infer1{A[\nu^\infty X. A/X] \vdash \nu^\infty X. A}
\infer2{A[B/X] \vdash \nu^\infty X. A}
\infer2{B \vdash \nu^\infty X. A}
\infer2{\vdash \nu^\infty X. A}
\end{prooftree}}
\]
This reduces to:
\[
\scalebox{0.7}{\begin{prooftree}
\hypo{w}
\infer1{\vdash B}
\hypo{u}
\infer1{B \vdash A[B/X]}
\infer2{\vdash A[B/X]}
\hypo{u^*}
\infer1{B \vdash \nu^\infty X. A}
\infer1[$A$]{A[B/X] \vdash A[\nu^\infty X.A/X]}
\infer2{\vdash A[\nu^\infty X. A/X]}
\hypo{}
\infer1{b < a < \infty, A[\nu^b X. A/X] \vdash A[\nu^b X. A/X]}
\infer1{b < a < \infty, \nu^\infty X. A \vdash A[\nu^b X. A/X]}
\infer1{a < \infty, \nu^\infty X. A \vdash \nu^a X. A}
\infer1[$A$]{a < \infty, A[\nu^\infty X. A/X] \vdash A[\nu^a X. A/X]}
\infer1{A[\nu^\infty X. A/X] \vdash \nu^\infty X. A}
\infer2{\vdash \nu^\infty X. A}
\end{prooftree}}
\]

By Proposition \ref{p:A-lemma-nu} together with the induction hypothesis on $\rho$, this is equivalent to:
\[
\scalebox{0.7}{\begin{prooftree}
\hypo{w}
\infer1{\vdash B}
\hypo{u}
\infer1{B \vdash A[B/X]}
\infer2{\vdash A[B/X]}
\hypo{\widetilde{u}}
\infer1{B \vdash \nu^\infty X. A}
\infer1[$A$]{A[B/X] \vdash A[\nu^\infty X.A/X]}
\infer2{\vdash A[\nu^\infty X. A/X]}
\hypo{}
\infer1{b < a < \infty, A[\nu^b X. A/X] \vdash A[\nu^b X. A/X]}
\infer1{b < a < \infty, \nu^\infty X. A \vdash A[\nu^b X. A/X]}
\infer1{a < \infty, \nu^\infty X. A \vdash \nu^a X. A}
\infer1[$A$]{a < \infty, A[\nu^\infty X. A/X] \vdash A[\nu^a X. A/X]}
\infer1{A[\nu^\infty X. A/X] \vdash \nu^\infty X. A}
\infer2{\vdash \nu^\infty X. A}
\end{prooftree}}
\]

But this proof is equivalent to $\mcut(w,  \widetilde{u})$, since we can reach it by performing the same reductions on $\mcut(w,x, i^{-1} \circ A \widetilde{u} \circ u)$.

\section{Conclusion}

This paper was motivated by an attempt to better understand the problem of compositionality in non-wellfounded proof theory in connection with Curry-Howard correspondence, especially in connection with the ``bouncing threads'' validity condition recently suggested in \cite{baelde2022bouncing}. For this purpose we made use of an approach to non-wellfounded proofs in terms of ordinal variables, stemming from the work of Dam, Gurov and Sprenger. We provided a computational interpretation of the proof system in terms of a computability predicate, and showed that the validity condition in terms of infinite descending chains of ordinal variables guarantees computability. As a consequence, we obtained a normalization result for proofs of finitary formulas, and in particular we obtained that proofs of sequents of the appropriate form uniquely represent total functions on natural numbers. Finally, we considered a categorical model based on the notion of computability, and showed that least and greatest fixpoints correspond to initial algebras and final coalgebras respectively. 

There are several tasks for future work, the most important of which is to settle the open question whether each bouncing thread valid proof can be represented by a proof using ordinal variables. Another important problem is cut elimination. In this paper, we did not consider ``infinitary proof theory'' in the sense of a full cut elimination theorem as in \cite{baelde2016infinitary,baelde2022bouncing,saurin2023linear}. The known proofs of these results are quite intricate, and rely on semantic arguments that are rather different from familiar techniques. An interesting direction of future research would be to revisit these results in the setting with ordinal variables, and to see whether this approach could help to clarify, simplify or elucidate the known proofs. 

\section*{Acknowledgement}
I wish to thank the two anonymous reviewers of an earlier draft of this paper for their helpful comments which led to several improvements.

\bibliographystyle{alphaurl}
\bibliography{my_bib}

\newcommand{\etalchar}[1]{$^{#1}$}
\begin{thebibliography}{AEL{\etalchar{+}}23}

\bibitem[AEL{\etalchar{+}}23]{afshari2023proof}
Bahareh Afshari, Sebastian Enqvist, Graham~E Leigh, Johannes Marti, and Yde Venema.
\newblock Proof systems for two-way modal mu-calculus.
\newblock {\em The Journal of Symbolic Logic}, pages 1--50, 2023.

\bibitem[AEL24]{afshari2024cyclic}
Bahareh Afshari, Sebastian Enqvist, and Graham~E Leigh.
\newblock Cyclic proofs for the first-order $\mu$-calculus.
\newblock {\em Logic Journal of the IGPL}, 32(1):1--34, 2024.

\bibitem[AP13]{abel2013wellfounded}
Andreas~M Abel and Brigitte Pientka.
\newblock Wellfounded recursion with copatterns: A unified approach to termination and productivity.
\newblock {\em ACM SIGPLAN Notices}, 48(9):185--196, 2013.

\bibitem[Bae12]{baelde2012least}
David Baelde.
\newblock Least and greatest fixed points in linear logic.
\newblock {\em ACM Transactions on Computational Logic (TOCL)}, 13(1):1--44, 2012.

\bibitem[BDKS22]{baelde2022bouncing}
David Baelde, Amina Doumane, Denis Kuperberg, and Alexis Saurin.
\newblock Bouncing threads for circular and non-wellfounded proofs: Towards compositionality with circular proofs.
\newblock In {\em Proceedings of the 37th Annual ACM/IEEE Symposium on Logic in Computer Science}, LICS '22, New York, NY, USA, 2022. Association for Computing Machinery.
\newblock \href {https://doi.org/10.1145/3531130.3533375} {\path{doi:10.1145/3531130.3533375}}.

\bibitem[BDS16]{baelde2016infinitary}
David Baelde, Amina Doumane, and Alexis Saurin.
\newblock Infinitary proof theory: the multiplicative additive case.
\newblock In {\em 25th EACSL Annual Conference on Computer Science Logic (CSL 2016)}, pages 42--1. Schloss Dagstuhl--Leibniz-Zentrum f{\"u}r Informatik, 2016.

\bibitem[BS11]{brotherston2011sequent}
James Brotherston and Alex Simpson.
\newblock Sequent calculi for induction and infinite descent.
\newblock {\em Journal of Logic and Computation}, 21(6):1177--1216, 2011.

\bibitem[BT17]{berardi2017equivalence}
Stefano Berardi and Makoto Tatsuta.
\newblock Equivalence of inductive definitions and cyclic proofs under arithmetic.
\newblock In {\em 2017 32nd Annual ACM/IEEE Symposium on Logic in Computer Science (LICS)}, pages 1--12. IEEE, 2017.

\bibitem[BT19]{berardi2019classical}
Stefano Berardi and Makoto Tatsuta.
\newblock Classical system of {M}artin-{L}\"{o}f's inductive definitions is not equivalent to cyclic proofs.
\newblock {\em Logical Methods in Computer Science}, 15, 2019.

\bibitem[CD23]{curzi2023computational}
Gianluca Curzi and Anupam Das.
\newblock Computational expressivity of (circular) proofs with fixed points.
\newblock In {\em 2023 38th Annual ACM/IEEE Symposium on Logic in Computer Science (LICS)}, pages 1--13. IEEE, 2023.

\bibitem[Cla09]{clairambault2009least}
Pierre Clairambault.
\newblock Least and greatest fixpoints in game semantics.
\newblock In Luca de~Alfaro, editor, {\em Foundations of Software Science and Computational Structures}, pages 16--31, Berlin, Heidelberg, 2009. Springer Berlin Heidelberg.

\bibitem[Das20a]{das2020circular}
Anupam Das.
\newblock A circular version of {G}\"{o}del's {T} and its abstraction complexity.
\newblock {\em arXiv preprint arXiv:2012.14421}, 2020.

\bibitem[Das20b]{das2020logical}
Anupam Das.
\newblock On the logical complexity of cyclic arithmetic.
\newblock {\em Logical Methods in Computer Science}, 16, 2020.

\bibitem[DG02]{dam2002mu}
Mads Dam and Dilian Gurov.
\newblock $\mu$-calculus with explicit points and approximations.
\newblock {\em Journal of Logic and Computation}, 12(2):255--269, 2002.

\bibitem[DM23]{das2023cyclic}
Anupam Das and Lukas Melgaard.
\newblock Cyclic proofs for arithmetical inductive definitions.
\newblock {\em arXiv preprint arXiv:2306.08535}, 2023.

\bibitem[EJS25]{ehrhard2025denotation}
Thomas Ehrhard, Farzad Jafarrahmani, and Alexis Saurin.
\newblock On the denotation of circular and non-wellfounded proofs in linear logic with fixed points.
\newblock In {\em 2025 40th Annual ACM/IEEE Symposium on Logic in Computer Science (LICS)}, pages 84--97. IEEE, 2025.

\bibitem[FS13]{fortier2013cuts}
J{\'e}r{\^o}me Fortier and Luigi Santocanale.
\newblock Cuts for circular proofs: semantics and cut-elimination.
\newblock In {\em Computer Science Logic 2013 (CSL 2013)}, pages 248--262. Schloss Dagstuhl--Leibniz-Zentrum fuer Informatik, 2013.

\bibitem[Jac99]{jacobs1999categorical}
Bart Jacobs.
\newblock {\em Categorical logic and type theory}, volume 141.
\newblock Elsevier, 1999.

\bibitem[Kna28]{knaster1928theoreme}
Bronis{\l}aw Knaster.
\newblock Un theoreme sur les functions d'ensembles.
\newblock {\em Ann. Soc. Polon. Math.}, 6:133--134, 1928.

\bibitem[Koz83]{kozen1983results}
Dexter Kozen.
\newblock Results on the propositional $\mu$-calculus.
\newblock {\em Theoretical computer science}, 27(3):333--354, 1983.

\bibitem[KS17]{kozen2017practical}
Dexter Kozen and Alexandra Silva.
\newblock Practical coinduction.
\newblock {\em Mathematical Structures in Computer Science}, 27(7):1132--1152, 2017.

\bibitem[NW96]{niwinski1996games}
Damian Niwi{\'n}ski and Igor Walukiewicz.
\newblock Games for the $\mu$-calculus.
\newblock {\em Theoretical Computer Science}, 163(1-2):99--116, 1996.

\bibitem[San02]{santocanale2002calculus}
Luigi Santocanale.
\newblock A calculus of circular proofs and its categorical semantics.
\newblock In {\em International Conference on Foundations of Software Science and Computation Structures}, pages 357--371. Springer, 2002.

\bibitem[Sau23]{saurin2023linear}
Alexis Saurin.
\newblock A linear perspective on cut-elimination for non-wellfounded sequent calculi with least and greatest fixed-points.
\newblock In {\em International Conference on Automated Reasoning with Analytic Tableaux and Related Methods}, pages 203--222. Springer, 2023.

\bibitem[SD03a]{sprenger2003global}
Christoph Sprenger and Mads Dam.
\newblock On global induction mechanisms in a $\mu$-calculus with explicit approximations.
\newblock {\em RAIRO-Theoretical Informatics and Applications}, 37(4):365--391, 2003.

\bibitem[SD03b]{sprenger2003structure}
Christoph Sprenger and Mads Dam.
\newblock On the structure of inductive reasoning: Circular and tree-shaped proofs in the $\mu$calculus.
\newblock In {\em International Conference on Foundations of Software Science and Computation Structures}, pages 425--440. Springer, 2003.

\bibitem[Sha14]{shamkanov2014circular}
Daniyar~Salkarbekovich Shamkanov.
\newblock Circular proofs for the {G}{\"o}del-{L}{\"o}b provability logic.
\newblock {\em Mathematical Notes}, 96:575--585, 2014.

\bibitem[Sim17]{simpson2017cyclic}
Alex Simpson.
\newblock Cyclic arithmetic is equivalent to {P}eano arithmetic.
\newblock In {\em International Conference on Foundations of Software Science and Computation Structures}, pages 283--300. Springer, 2017.

\bibitem[SS21]{savateev2021non}
Yury Savateev and Daniyar Shamkanov.
\newblock Non-well-founded proofs for the {G}rzegorczyk modal logic.
\newblock {\em The Review of Symbolic Logic}, 14(1):22--50, 2021.

\bibitem[SU06]{sorensen2006lectures}
Morten~Heine S{\o}rensen and Pawel Urzyczyn.
\newblock {\em Lectures on the Curry-Howard isomorphism}, volume 149.
\newblock Elsevier, 2006.

\bibitem[Tar55]{tarski1955lattice}
Alfred Tarski.
\newblock A lattice-theoretical fixpoint theorem and its applications.
\newblock {\em Pacific J. Math}, 5:285--309, 1955.

\bibitem[Wal00]{walukiewicz2000completeness}
Igor Walukiewicz.
\newblock Completeness of {K}ozen's axiomatisation of the propositional $\mu$-calculus.
\newblock {\em Information and Computation}, 157(1-2):142--182, 2000.

\end{thebibliography}

\end{document}